\documentclass[UKenglish,cleveref, autoref, dvipsnames]{lipics-v2019}
\usepackage{colortbl}
\usepackage{microtype}
\usepackage{diagbox}
\usepackage{amsmath}
\usepackage{amsfonts}
\usepackage{amssymb}
\usepackage{makeidx}
\usepackage{graphicx}
\usepackage{adjustbox}
\usepackage{mathabx}
\usepackage[inline]{enumitem}
\usepackage{macros}
\usepackage{amsthm}
\usepackage{xspace}

\usepackage{verbatim}

\usepackage{wrapfig}

\usepackage{tikz}
\usetikzlibrary{backgrounds,automata}
\usetikzlibrary{shapes,snakes}
\usetikzlibrary{arrows,decorations.pathmorphing}

\hideLIPIcs
\nolinenumbers

\newcommand{\outcome}{{\sf Out}}
\renewcommand{\val}{{\sf Val}}
\newcommand{\arena}{{\mathcal{A}}}
\newcommand{\game}{{\mathcal{G}}}
\newcommand{\MP}{{\sf MP}}
\newcommand{\MPsup}{\overline{{\sf MP}}}
\newcommand{\MPinf}{\underline{{\sf MP}}}
\newcommand{\BR}{{\sf BR}}
\newcommand{\Estrat}[1]{{\ll#1\gg~}}
\newcommand{\ASV}{{\sf ASV}}
\newcommand{\CSV}{{\sf CSV}}
\newcommand{\WCV}{{\sf WCV  }}
\renewcommand{\last}{{\it last  }}
\newcommand{\DS}{{\sf DS}}

\newcommand*\circled[1]{\tikz[baseline=(char.base)]{
            \node[shape=circle,draw,inner sep=1pt] (char) {#1};}}

\title{The Adversarial Stackelberg Value in Quantitative
  Games}
\date{}

\author{Emmanuel Filiot}{Universit\'e libre de Bruxelles (ULB), Belgium}{}{}{}

\author{Raffaella Gentilini}{University of Perugia, Italy}{}{}{}

\author{Jean-Fran\c{c}ois Raskin}{Universit\'e libre de Bruxelles (ULB), Belgium}{}{}{}

\acknowledgements{We thank anonymous reviewers, Dr. Shibashis Guha and Ms. Mrudula Balachander for useful comments on a preliminary version of this paper. This work is partially supported by the PDR project
  Subgame perfection in graph games (F.R.S.-FNRS), the MIS project
  F451019F (F.R.S.-FNRS), the ARC project Non-Zero Sum Game Graphs:
  Applications to Reactive Synthesis and Beyond (Fédération
  Wallonie-Bruxelles), the EOS project Verifying Learning Artificial
  Intelligence Systems (F.R.S.-FNRS \& FWO), and the COST Action 16228
  GAMENET (European Cooperation in Science and Technology). Emmanuel
  Filiot is associate researcher at F.R.S.-FNRS.}

\authorrunning{Emmanuel Filiot, Raffaella Gentillini, and Jean-Fran\c{c}ois Raskin}

\Copyright{E. Filiot, R. Gentillini, and J.-F. Raskin}



\def\out{\sf{Out}}

\renewcommand*{\qed}{\hfill\ensuremath{\blacksquare}}%

\begin{document}

\maketitle

\begin{abstract}
In this paper, we study the notion of adversarial Stackelberg value for two-player non-zero sum games played on bi-weighted graphs with the mean-payoff and the discounted sum functions. The adversarial Stackelberg value of Player 0 is the largest value that Player 0 can obtain when announcing her strategy to Player 1 which in turn responds with any of his best response. For the mean-payoff function, we show that the adversarial Stackelberg value is not always achievable but $\epsilon$-optimal strategies exist. We show how to compute this value and prove that the associated threshold problem is in {\sf NP}. For the discounted sum payoff function, we draw a link with the target discounted sum problem which explains why the problem is difficult to solve for this payoff function. We also provide solutions to related gap problems.
\end{abstract}
\newpage
\section{Introduction}

In this paper, we study two-player non-zero sum infinite duration quantitative games played on graph games. In non-zero sum games, the notion of worst-case value is not rich enough to reason about the (rational) behavior of players. More elaborate solution concepts have been proposed in game theory to reason about non-zero sum games: Nash equilibria, subgames perfect equilibria, admissibility, and Stackelberg equilibria are important examples of such solution concepts, see e.g.~\cite{nash50} and \cite{osbornebook}.

Let us first recall the abstract setting underlying the notion of Stackelberg equilibria and explain the variant that is the focus of this paper. Stackelberg games are strategic games played by two players. We note $\Sigma_0$ the set of strategies of Player~0, also called the {\em leader}, and $\Sigma_1$ the set of strategies of Player~1, also called the {\em follower}. Additionally, the game comes with two (usually $\mathbb{R}$-valued) payoff functions, ${\sf Payoff}_0$ and ${\sf Payoff}_1$, that determine the payoff each player receives: if $\sigma_0 \in \Sigma_0$ and $\sigma_1 \in \Sigma_1$ are chosen then Player~0 receives the payoff ${\sf Payoff}_0(\sigma_0,\sigma_1)$ while Player~1 receives the payoff ${\sf Payoff}_1(\sigma_0,\sigma_1)$. Both players aim at maximizing their respective payoffs, and in a Stackelberg game, players play sequentially as follows.
      \circled{1} Player~0, the leader, announces her choice of strategy $\sigma_0 \in \Sigma_0$.
      \circled{2} Player~1, the follower, announces his choice of strategy $\sigma_1 \in \Sigma_1$.
      \circled{3} Both players receive their respective payoffs: ${\sf Payoff}_0(\sigma_0,\sigma_1)$ and ${\sf Payoff}_1(\sigma_0,\sigma_1)$.
Due to the sequential nature of the game, Player~1 knows the strategy $\sigma_0$, and so to act rationally (s)he should choose a strategy $\sigma_1$ that maximizes the payoff ${\sf Payoff}_1(\sigma_0,\sigma_1)$. If such a strategy $\sigma_1$ exists, it is called a {\em best-response}~\footnote{As we will see later in the paper, sometimes, best-responses are not guaranteed to exist. In such cases, we need to resort to weaker notions such as  $\epsilon$-best-responses. We leave those technical details for later in the paper.} to the strategy $\sigma_0 \in \Sigma_0$. In turn, if the leader assumes a rational response of the follower to her strategy, this should guide the leader when choosing $\sigma_0 \in \Sigma_0$. Indeed, the leader should choose a strategy $\sigma_0 \in \Sigma_0$ such that the value ${\sf Payoff_0}(\sigma_0,\sigma_1)$ is as large as possible when $\sigma_1$ is a best-response of the follower. 

Two different scenarios can be considered in this setting: either the best-response $\sigma_1 \in \Sigma_1$ is imposed by the leader (or equivalently chosen {\em cooperatively} by the two players), or the best-response is chosen {\em adversarially} by Player~1. In classical results from game theory and most of the close related works on games played on graphs~\cite{DBLP:conf/time/GuptaS14,DBLP:journals/corr/GuptaS14c}, with the exception of~\cite{DBLP:journals/amai/KupfermanPV16}, only the cooperative scenario has been investigated. But, the adversarial case is interesting because it allows us to model the situation in which the leader chooses $\sigma_0 \in \Sigma_0$ only and must be prepared to face any rational response of Player~1, i.e. if Player~1 has several possible best responses then $\sigma_0$ should be designed to face all of them. In this paper, our main contribution is to investigate the second route. As already noted in~\cite{DBLP:journals/amai/KupfermanPV16}, this route is particularly interesting for applications in {\em automatic synthesis}. Indeed, when designing a program, and this is especially true for reactive programs~\cite{DBLP:conf/popl/PnueliR89,DBLP:conf/lata/BrenguierCHPRRS16}, we aim for robust solutions that works for multiple rational usages, e.g. all the usages that respect some specification or that maximize some measure for the user. 

To reflect the two scenarios above, there are two notions of {\em
  Stackelberg values}. First, the {\em cooperative Stackelberg value}
is the largest value that Player~0 can secure by proposing a strategy
$\sigma_0$ and a strategy $\sigma_1$ to the follower with the
constraint that $\sigma_1$ is a best-response for the follower to
$\sigma_0$. Second, the {\em adversarial Stackelberg value} is the
largest value that Player~0 can secure by proposing a strategy
$\sigma_0$ and facing any best response $\sigma_1$ of the follower to
the strategy $\sigma_0$. In this paper, we mostly concentrate on the
{\em adversarial} Stackelberg value, for infinite duration games
played on bi-weighted game graphs for the mean-payoff function and the
discounted sum function. The cooperative case has been studied
in~\cite{DBLP:conf/time/GuptaS14,DBLP:journals/corr/GuptaS14c} and we
only provide some additional results when relevant for that case (see
also related works below).


\vspace{-3mm}
\subparagraph{{\bf Main contributions}} First, we consider the mean-payoff function. For this payoff function, best responses of Player~1 to a strategy $\sigma_0 \in \Sigma_0$ not always exist (Lemma~\ref{lem:emptiness-nonemptiness}). As a consequence, the {\em cooperative} (\CSV) and {\em adversarial} (\ASV) 
Stackelberg values are defined using $\epsilon$-best responses.
While strategies of Player~0 to achieve $\CSV$ always exist as shown in~\cite{DBLP:conf/time/GuptaS14}, we show that it is not the case for $\ASV$ (Theorem~\ref{thm:equiv-def}). The $\ASV$ can only be approached in general and memory may be necessary to play optimally or $\epsilon$-optimally in adversarial Stackelberg games for the mean-payoff function (Theorem~\ref{thm:equiv-def}). We also provide results for related algorithmic problems. 
We provide a notion of witness for proving that the $\ASV$ is (strictly) above some threshold (Theorem~\ref{thm:witness-mp}), and it is the basis for an {\sf NP} algorithm to solve the threshold problem (Theorem~\ref{NP-membership-ASV-MP}). Finally, we show how the $\ASV$ can be computed effectively (Theorem~\ref{thm:effectiveASV}).

Second, we consider the discounted sum function. In that case, best responses of Player~1 to strategies $\sigma_0 \in \Sigma_0$ of Player~0 always exist (Lemma~\ref{lem:exists-ds-br}). The $\CSV$ and $\ASV$ are directly based on best-responses in that case. Then we draw a link between the {\em target discounted sum problem} and the $\CSV$ threshold problem (Lemma~\ref{lem:target-csv}). The target discounted sum problem has been studied recently in~\cite{DBLP:conf/lics/BokerHO15}, left open there for the general case and shown to be related to other open problems in piecewise affine maps and the representation of numbers in nonintegral bases. As a consequence, we introduce a relaxation of the threshold problems for both $\CSV$ and $\ASV$ in the form of gap problems (or promised problems as defined in~\cite{DBLP:conf/birthday/Goldreich06a}). We provide algorithms to solve those gap problems (Theorem~\ref{thm:ASV-CSV-gap}) both for $\CSV$ and $\ASV$. Finally, we prove {\sf NP}-hardness for the gap problems both for $\CSV$ and $\ASV$ (Theorem~\ref{thm:hardness-gap-ds}).

\vspace{-3mm}
\subparagraph{{\bf Closely related work}}
The notions of cooperative and adversarial synthesis have been introduced in~\cite{DBLP:conf/tacas/FismanKL10,DBLP:journals/amai/KupfermanPV16}, and further studied in~\cite{DBLP:conf/icalp/ConduracheFGR16,DBLP:conf/lics/FiliotGR18}. Those two notions are closely related to our notion of cooperative and adversarial Stackelberg value respectively. The games that are considered in those papers are infinite duration games played on graphs but they consider Boolean $\omega$-regular payoff functions or finite range $\omega$-regular payoff functions. Neither the mean-payoff function nor the discounted sum payoff function are $\omega$-regular, and thus they are not considered in~\cite{DBLP:conf/tacas/FismanKL10,DBLP:journals/amai/KupfermanPV16}. The $\omega$-regularity of the payoff functions that they consider is central to their techniques: they show how to reduce their problems to problems on tree automata and strategy logic. Those reductions cannot be used for payoff functions that are {\em not} $\omega$-regular functions and we need specific new techniques to solve our problems.

In~\cite{DBLP:conf/time/GuptaS14,DBLP:journals/corr/GuptaS14c}, the
cooperative scenario for Stackelberg game is studied for mean-payoff
and discounted sum respectively. Their results are sufficient to solve
most of the relevant questions on the $\CSV$ but not for
$\ASV$. Indeed, the techniques that are used for $\CSV$ are closely
related to the techniques that are used to reason on Nash equilibria
and build on previous works~\cite{DBLP:conf/lfcs/BrihayePS13} which in
turn reduce to algorithmic solutions for zero-sum one dimensional mean-payoff (or discounted sum games). For the $\ASV$ in the context of the mean-payoff function, we have to use more elaborate
multi-dim. mean-payoff games and a notion of Pareto curve adapted
from~\cite{DBLP:conf/cav/BrenguierR15}. Additionally, we provide new
results on the $\CSV$ for the discounted sum function. First, our
reduction that relates the target discounted sum problem to the $\CSV$ is new and gives additional explanations why the $\CSV$ is difficult to solve and not solved in the general case in~\cite{DBLP:journals/corr/GuptaS14c}. Second, while we also leave the general problem open here, we show how to solve the gap problems related to both $\CSV$ and $\ASV$. Finally, the authors of~\cite{DBLP:conf/sefm/GuptaSTDP16} study {\em incentive equilibria} for multi-player mean-payoff games. This work is an extension of their previous work~\cite{DBLP:conf/time/GuptaS14} and again concentrates on $\CSV$ and does not consider $\ASV$.


\vspace{-3mm}\subparagraph{{\bf Structure of the paper}} In Sect.~2, we introduce the necessary preliminaries for our definitions and developments. In Sect.~3, we consider the adversarial Stackelberg value for the mean-payoff function. In Sect.~4, we present our results for the discounted sum function.

\section{Preliminaries and notations}

\subparagraph{{\bf Arenas}} A (bi-weighted) \emph{arena}  $\mathcal{A}=(V, E, \langle V_0, V_1\rangle, w_0,w_1)$ consists of a finite set $V$ of vertices, a set $E\subseteq V\times V$ of edges such that for all $v\in V$ there exists   $v'\in V$ such that    $(v,v')\in E$, a partition  $\langle V_0,V_1\rangle$ of $V$, where $V_0$ (resp. $V_1$) is   the set of vertices for Player $0$ (resp. Player $1$), and two edge weight functions $w_0:E \mapsto \mathbb{Z},w_1: E\mapsto \mathbb{Z}$. In the sequel, we denote the maximum absolute value of a weight in ${\cal A}$ by $W$. As arenas are directed weighted graphs, we use, sometimes without recalling the details, the classical vocabulary for directed graphs. E.g., a set of vertices $S \subseteq V$ is a strongly connected component of the arena ({\sf SCC} for short), if for all $s_1,s_2 \in S$, there exists a path from $s_1$ to $s_2$ and a path from $s_2$ to $s_1$.

\subparagraph{{\bf Plays and histories}}
 A \emph{play} in $\mathcal{A}$ is an infinite sequence of vertices $\pi=\pi_0\pi_1\dots \in V^\omega$ such that for all $k \in  \mathbb{N}, (\pi_k,\pi_{k+1}) \in E$. We denote  by $\sf{Plays}_\mathcal{A}$  the set of plays  in $\mathcal{A}$, omitting the subscript $\mathcal{A}$ when the underlying arena is clear from the context. 
Given  $\pi=\pi_0\pi_1\dots \in \sf{Plays}_\mathcal{A}$ and $k\in \mathbb{N}$, the prefix $\pi_0\pi_1\dots \pi_k$ of $\pi$ (resp. suffix $\pi_k\pi_{k+1}\dots$ of $\pi$) is denoted by $\pi_{\leq k}$  (resp.  $\pi_{\geq k}$). 
 An \emph{history} in $\mathcal{A}$ is a (non-empty) prefix of a play in $\mathcal{A}$. The length $|h|$ of an history  $h=\pi_{\leq k} $ is the number $|h|=k$  of its edges. We denote  by $\sf{Hist}_\mathcal{A}$ the set of  histories in $\mathcal{A}$, $\mathcal{A}$ is omitted when clear from the context. Given $i\in \{0,1\}$, the set $\sf{Hist}^{\it i}_\mathcal{A}$ denotes   the set of histories such that their last vertex  belongs to $V_i$. We denote the last vertex of a history $h$ by ${\sf last}(h)$.
 We write $h \leq \pi$ whenever $h$ is a prefix of $\pi$.
A play $\pi$ is called a lasso if it is obtained as the concatenation of 
a history $h$ concatenated with the infinite repetition of another
history $l$, i.e. $\pi = h \cdot l^\omega$ with $h,l \in
\sf{Hist}_\mathcal{A}$ (notice that $l$ is not necessary a simple
cycle). The {\em size} of a lasso $h \cdot l^\omega$ is defined as $|h \cdot l|$.
Given a vertex $v\in V$ in the  arena $\mathcal{A}$, we denote by $Succ(v) = \{v' | (v,v') \in E\}$ the set of successors of $v$  and by $Succ^*$ its transitive closure.

\subparagraph{{\bf Games}} A \emph{game} $\mathcal{G}=(\mathcal{A}, \langle \val_0,\val_1\rangle )$ consists of a bi-weighted arena $\mathcal{A}$, a value function  $\val_0: \sf{Plays}_\mathcal{A}\mapsto \mathbb{R} $    for Player $0$ and a value function $\val_1: \sf{Plays}_\mathcal{A}\mapsto \mathbb{R} $ for  Player $1$. In this paper, we consider the classical \emph{mean-payoff} and \emph{discounted-sum} value functions. Both are played in bi-weighted arenas.

 In a \emph{mean-payoff} game $\mathcal{G}=(\mathcal{A},   \langle \sf{MP}_0,\sf{MP}_1\rangle )$   the payoff functions  $\sf{MP}_0,  \sf{MP}_1$ are defined as follows. Given a play  $\pi\in\sf Plays_\mathcal{A}$ and  $i\in\{0,1\}$, the payoff  $\MPinf_i(\pi)$ is given by $\MPinf_i(\pi)=\liminf_{k\rightarrow\infty}\frac{1}{k}  w_i(\pi_{\leq k}) $, where the weight $w_i(h)$ of an history $h\in\sf{Hist}$ is the sum of the weights assigned by $w_i$ to its edges. In our definition of the mean-payoff, we have used $\liminf$, we will also need the $\limsup$ case for technical reasons. Here is the formal definition together with its notation:  $\MPsup_i(\pi)=\limsup_{k\rightarrow\infty}\frac{1}{k}  w_i(\pi_{\leq k}) $
 
 For a given discount factor $0<\lambda<1$, a  {\em discounted sum}
 game is a game  $\mathcal{G}=(\mathcal{A},   \langle
 \sf{DS}^\lambda_0,\sf{DS}^\lambda_1\rangle )$ where  the payoff
 functions  $\sf{DS}^\lambda_0,  \sf{DS}^\lambda_1$ are defined as
 follows. Given a play  $\pi\in\sf Plays_\mathcal{A}$ and
 $i\in\{0,1\}$, the payoff  $\sf{DS}^\lambda_i(\pi)$ is defined as ${\sf{DS}}^\lambda_i(\pi)= \sum_{k=0}^\infty  \lambda^k w_i(\pi_k,\pi_{k+1}) $.

\subparagraph{{\bf Strategies and payoffs}} A strategy for Player $i\in \{0,1\}$  in a game $\mathcal{G}=(\mathcal{A}, \langle \val_0,\val_1\rangle )$ is a function $\sigma:{\sf{Hist}}^i_\mathcal{A}\mapsto V$ that maps histories ending with a vertex $v\in V_i$  to a successor of $v$. The set of all strategies of Player $i\in\{0,1\}$ in the game $\mathcal{G}$ is denoted $\Sigma_i(\mathcal{G})$, or $\Sigma_i$ when $\mathcal{G}$ is clear from the context.

A strategy has memory $M$ if it can be realized as the output of a finite state machine with $M$ states. A memoryless (or positional) strategy is a strategy with memory $1$, that is, a function that only depends on the last element of the given partial play. We note $\Sigma_i^{\sf ML}$ the set of memoryless strategies of Player~$i$, and $\Sigma_i^{\sf FM}$ its set of finite memory strategies.
 A \emph{profile} is a pair of strategies $\bar{\sigma}=(\sigma_0,\sigma_1)$, where $\sigma_0\in \Sigma_0(\mathcal{G})$ and $\sigma_1\in \Sigma_1(\mathcal{G})$.
  As we consider games with perfect information and deterministic transitions, any  profile $\bar{\sigma}$  yields, from any history $h$, a unique {\em play} or \emph{outcome}, denoted ${\sf Out}_h(\mathcal{G}, \bar{\sigma} )$. Formally, ${\sf Out}_h(\mathcal{G}, \bar{\sigma} )$ is the play $\pi$ such that $\pi_{\leq |h|-1} = h$  and $\forall k\geq |h|-1$ it holds that $\pi_{k+1} =\sigma_i(\pi_{\leq k})$ if $\pi_k \in V_i$. The set of outcomes (resp. histories) compatible with a strategy $\sigma\in  \Sigma_{i\in\{0,1\}}(\mathcal{G})$ after a history $h$ is ${\sf Out}_h(\mathcal{G}, \sigma) = \{\pi \:|\: \exists \sigma'\in\Sigma_{1-i}(\mathcal{G}) \mbox{ such that } \pi = {\sf Out}_h(\mathcal{G}, (\sigma, \sigma'))\}$ (resp. ${\sf Hist}_h(\sigma) = \{h' \in {\sf Hist}(\mathcal{G}) \:|\: \exists \pi \in {\sf Out}_h(\mathcal{G}, \sigma), n\in\mathbb{N}: h' = \pi_{\leq n}\}$.

  Each outcome $\pi$ in  $\mathcal{G}=(\mathcal{A}, \langle \val_0,\val_1\rangle )$ yields a payoff $\val(\pi)=(\val_0(\pi),\val_1(\pi))$, where $\val_0(\pi)$ is the payoff  for Player $0$ and  $\val_1(\pi)$ is the payoff  for Player $1$. We denote by $\val(h,\bar{\sigma})=\val({\sf Out}_h(\mathcal{G}, \bar{\sigma} ))$ the payoff of a profile of strategies $\bar{\sigma}$
 after a history $h$.
 
 Usually, we consider game instances such that players start to play at a fixed vertex $v_0$.
Thus, we call an initialized game a pair $(\mathcal{G}, v_0)$, where  $\mathcal{G}$ is a game and $v_0\in V$ is the initial vertex. When the initial vertex $v_0$ is clear from context, we speak directly of $\mathcal{G}, \outcome(\mathcal{G}, \bar{\sigma} ), \outcome(\mathcal{G}, \sigma) $, $\val(\bar{\sigma})$ instead of $(\mathcal{G}, v_0)$, $\outcome_{v_0} (\mathcal{G}, \bar{\sigma} )$, $\outcome_{v_0}(\mathcal{G}, \sigma)$,$\val(v_0,\bar{\sigma})$ . We sometimes simplify further the notation omitting also $\mathcal{G}$, when the latter is clear from the context.

\vspace{-3mm}
\subparagraph{{\bf Best-responses and adversarial value in zero-sum games}} 
Let $\game=(\mathcal{A}, \langle \val_0, \val_1 \rangle)$ be a $(\val_0,\val_1)$-game on the bi-weighted arena $\arena$. Given a strategy $\sigma_0$ for Player~0, we define two sets of strategies for Player~1. His {\em best-responses} to $\sigma_0$, noted $\BR_1(\sigma_0)$, and defined as: 
            $$\left\{  \sigma_1 \in \Sigma_1 \mid \forall v \in V \cdot  \forall \sigma'_1 \in \Sigma_1: \val_1(\outcome_v(\sigma_0,\sigma_1)) \geq \val_1(\outcome_v(\sigma_0,\sigma'_1)) \right\}.$$ 
\noindent 
And his  {\em $\epsilon$-best-responses} to $\sigma_0$, for $\epsilon >0$, noted $\BR^{\epsilon}_1(\sigma_0)$, and defined as:
            $$\left\{  \sigma_1 \in \Sigma_1 \mid \forall v \in V \cdot \forall \sigma'_1 \in \Sigma_1: \val_1(\outcome_v(\sigma_0,\sigma_1) \geq \val_1(\outcome_v(\sigma_0,\sigma'_1))-\epsilon \right\}.$$ 
  \noindent
  We also introduce notations for zero-sum games (that are needed as intermediary steps in our algorithms).
  The adversarial value that Player~1 can enforce in the game $\game$ from vertex $v$ as:
  $\WCV_1(v)=\sup_{\sigma_1 \in \Sigma_1}\inf_{\sigma_0 \in \Sigma_0} \val_1(\outcome_v(\sigma_0,\sigma_1))$.
Let $\mathcal{A}$ be an arena, $v
\in V$ one of its states, and ${\cal O} \subseteq
\sf{Plays}_\mathcal{A}$ be a set of plays (called objective), then we
write $\mathcal{A}, v \models \Estrat{i} {\cal O}$, if $\exists
\sigma_i \in \Sigma_i \cdot \forall \sigma_{1-i} \in \Sigma_{1-i}:
{\sf Out}_v(\mathcal{A}, (\sigma, \sigma')) \in {\cal O}$, for $i \in
\{0,1\}$. Here the underlying interpretation is zero-sum: Player~$i$
wants to force an outcome in ${\cal O}$ and Player~$1-i$ has the
opposite goal. All the zero-sum games we consider in this paper are {\em determined} meaning that for all $\mathcal{A}$, for all objectives ${\cal O} \subseteq \sf{Plays}_\mathcal{A}$ we have that:
  $ \mathcal{A},v \models \Estrat{i} {\cal O}$ iff 
        $\mathcal{A},v \nvDash \Estrat{1-i} \sf{Plays}_\mathcal{A} \setminus {\cal O}$.

\subparagraph{Convex hull and $F_{\sf min}$}
First, we need som  e additional notations and vocabulary related to
linear algebra.
Given a finite set of $d$-dim. vectors $X \subset
\mathbb{R}^d$, we note the set of all their convex combinations as
${\sf CH}(X)=\{ v \mid v=\sum_{x \in X} \alpha_x \cdot x \land \forall
x \in X:\alpha_x \in [0,1] \land \sum_{x \in X} \alpha_x=1 \}$, this
set is called the {\em convex hull} of $X$. We also need the following additional, and less standard notions, introduced
in~\cite{DBLP:conf/concur/ChatterjeeDEHR10}. Given a finite set of
$d$-dim. vectors $X \subset \mathbb{R}^d$, let $f_{\sf min}(X)$ be the vector $v=(v_1,v_2,\dots,v_d)$ where $v_i=\min\ \{ c \mid \exists x \in X : x_i=c \}$, i.e. the vector $v$ is the pointwise minimum of the vectors in $X$.
Let $S \subseteq \mathbb{R}^d$, then $F_{\sf min}(S)=\{ f_{\sf min}(P) \mid P \mbox{~is a finite subset of~} S\}$. The following proposition expresses properties of the $F_{\sf min}(S)$ operator that are useful for us in the sequel. The interested reader will find more results about the $F_{\sf min}$ operator in~\cite{DBLP:conf/concur/ChatterjeeDEHR10}.

\begin{proposition}
\label{prop:larger}
For all sets $S \subseteq \mathbb{R}^d$, for all $x \in F_{\sf min}(S)$, there exists $y \in S$ such that $x \leq y$. If $S$ is a closed bounded set then $F_{\sf min}(S)$ is also a closed bounded set.
\end{proposition}

In the sequel, we also use formulas of the theory of the reals with addition and order, noted $\langle \mathbb{R},+,\leq \rangle$, in order to define subsets of $\mathbb{R}^n$. This theory is decidable and admits effective quantifier elimination~\cite{DBLP:journals/siamcomp/FerranteR75}.

\section{Adversarial Stackelberg value for mean-payoff games}
\label{secMP}

\subparagraph{{\bf Mean-payoffs induced by simple cycles}}
 Given a
play $\pi \in \sf{Plays}_\mathcal{A}$, we note ${\sf inf}(\pi)$ the
set of vertices $v$ that appear infinitely many times along $\pi$,
i.e. ${\sf inf}(\pi)=\{ v \mid \forall i \in \mathbb{N} \cdot \exists
j \geq i: v=\pi_j \}$. It is easy to see that ${\sf inf}(\pi)$ is an
{\sf SCC} in the underlying graph of the arena $\mathcal{A}$. A {\em
  cycle} $c$ is a sequence of edges that starts and stops in a given
vertex $v$, it is {\em simple} if it does not contain any other
 repetition of vertices. Given an {\sf SCC} $S$, we write
$\mathbb{C}(S)$ for the set of simple cycles inside $S$.  Given a
simple cycle $c$, for $i \in \{0,1\}$, let
$\MP_i(c)=\frac{w_i(c)}{|c|}$ be the mean of $w_i$ weights along edges
in the simple cycle $c$, and we call the pair $(\MP_0(c),\MP_1(c))$
the {\em mean-payoff coordinate} of the cycle $c$. We write ${\sf
  CH}(\mathbb{C}(S))$ for the convex-hull of the set of mean-payoff
coordinates of simple cycles of $S$.  The following theorem relates the $d$-dim. mean-payoff values of infinite plays
and the $d$-dim. mean-payoff of simple cycles in the arena.

\begin{theorem}[\cite{DBLP:conf/concur/ChatterjeeDEHR10}]
\label{thm:caractSCC}
Let $S$ be an {\sf SCC} in the arena $\mathcal{A}$, the following two properties hold:
      $(i)$ for all $\pi \in \sf{Plays}_\mathcal{A}$, if ${\sf inf}(\pi) \subseteq S$ then $(\MPinf_0(\pi),\MPinf_1(\pi)) \in F_{\sf min}({\sf CH}(\mathbb{C}(S)))$
      $(ii)$ for all $(x,y) \in F_{\sf min}({\sf CH}(\mathbb{C}(S)))$, there exists $\pi \in \sf{Plays}_\mathcal{A}$ such that ${\sf inf}(\pi)=S$ and $(\MPinf_0(\pi),\MPinf_1(\pi))=(x,y)$.
      Furthermore, the set $F_{\sf min}({\sf CH}(\mathbb{C}(S)))$ is effectively expressible in $\langle \mathbb{R},+,\leq \rangle$. 
\end{theorem}

In the sequel, we denote by $\Phi_{S}(x,y)$ the formula with two free variables in $\langle \mathbb{R},+,\leq \rangle$ such that for all $(u,v) \in \mathbb{R}^2$, $(u,v) \in F_{\sf min}({\sf CH}(\mathbb{C}(S)))$ if and only if $\Phi_{S}(x,y)[x/u,y/v]$ is true.



\subparagraph{{\bf On the existence of best-responses for $\MP$}}
We start the study of mean-payoff games with some considerations about the existence of best-responses and  $\epsilon$-best-responses for Player~1 to strategies of Player~0.

\begin{lemma} 
\label{lem:emptiness-nonemptiness}
      There is a mean-payoff game $\game$ and a strategy $\sigma_0 \in \Sigma_0(\game)$ such that $\BR_1(\sigma_0)=\varnothing$. 
     For all mean-payoff games $\game$ and finite memory strategies $\sigma_0 \in \Sigma^{\sf FM}_0(\game)$,  $\BR_1(\sigma_0)\neq\varnothing$.
     For all mean-payoff games $\game$, for all strategies $\sigma_0 \in \Sigma_0(\game)$, for all $\epsilon >0$, $\BR^{\epsilon}_1(\sigma_0)\neq\varnothing$.
\end{lemma}
\begin{proof}[Proof sketch - full proof in Appendix]
First, in the arena of Fig.~\ref{fig:game-empty}, we consider the strategy of Player~0 that plays the actions $c$ and $d$ with a frequency that is equal to $1-{\frac{1}{k}}$ for $c$ and $\frac{1}{k}$ for $d$ where $k$ is the number of times that Player~1 has played $a$ in state $1$ before sending the game to state $2$. We claim that there is no best response of Player~1 to this strategy of Player~0. Indeed, taking $a$ one more time before going to state $2$ is better for Player~1. 

Second, if Player~0 plays a finite memory strategy, then a best response for Player~1 is an optimal path for the mean-payoff of Player~1 in the finite graph obtained as the product of the original game arena with the finite state strategy of Player~0. Optimal mean-payoff paths are guaranteed to exist~\cite{DBLP:journals/dm/Karp78}. 

Finally, the existence of $\epsilon$-best responses for $\epsilon > 0$, is guaranteed by an analysis of the infinite tree obtained as the unfolding of the game arena with the (potentially infinite memory) strategy of Player~1. Branches of this tree witness responses of Player~1 to the strategy of Player~0. The supremum of the values of those branches for Player~1 is always approachable to any $\epsilon > 0$.
\end{proof}

\begin{figure}[t]
    \centering

\begin{tikzpicture}[->,>=stealth',shorten >=1pt,auto,node distance=2.8cm,
                    thick,scale=0.7,every node/.style={scale=0.8}]
  \tikzstyle{every state}=[fill=gray!30,text=black]
  \tikzstyle{every edge}=[draw=black]
  \tikzstyle{initial}=[initial by arrow]


	\node[thick,state,initial left,initial text=,rectangle,minimum size=25pt]
        (s1) at (0,0) {$1$};

	\node[thick,state,minimum size=25pt]
        (s2) at (4,0) {$2$};

	\node[thick,state,minimum size=25pt]
        (s3) at (8,0) {$3$};

        \draw (s1) edge [loop above] node {(0,0)} node [left,xshift=-0.1cm,yshift=-0.35cm] {$a$} (s1) ;

        \draw (s1) edge node {(0,0)} node [below] {$b$} (s2) ;

        \draw (s2) edge [loop above] node {(0,1)} node [left,xshift=-0.1cm,yshift=-0.35cm] {$d$} (s2) ;

        \draw (s3) edge [loop above] node {(0,2)} node [left,xshift=-0.1cm,yshift=-0.35cm] {$c$} (s3) ;

        \draw (s2) edge [bend left] node [yshift=-0.07cm] {(0,2)} node [below] {$c$} (s3) ;

        \draw (s3) edge [bend left] node [above,yshift=-0.07cm] {(0,1)} node [below] {$d$} (s2) ;
\end{tikzpicture}

    \caption{A mean-payoff game in which there exists a
      Player~0's strategy $\sigma_0$ such that $\BR_1(\sigma_0)=\varnothing$.}
    \label{fig:game-empty}
\vspace{-5mm}
\end{figure}
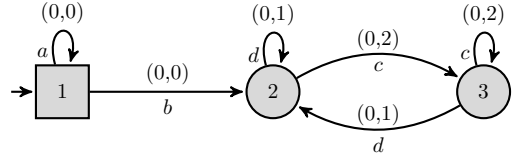

According to Lemma~\ref{lem:emptiness-nonemptiness}, the set of best-responses of Player~1 to a strategy of Player~0 can be empty. As a consequence, we need to use the notion of  $\epsilon$-best-responses (which are always guaranteed to exist) when we define the adversarial Stackelberg value:
$$\ASV(\sigma_0)(v)=\sup_{\epsilon \geq 0 ~\mid~ \BR_{1}^{\epsilon}(\sigma_0)\neq\varnothing} ~~~\inf_{\sigma_1 \in \BR_{1}^{\epsilon}(\sigma_0)} \MPinf_0(\outcome_v(\sigma_0,\sigma_1))
\mbox{~and~}\ASV(v)=\sup_{\sigma_0 \in \Sigma_0} \ASV(\sigma_0)(v)
$$
\noindent
We note that when best-responses to a strategy $\sigma_0$ exist, then
as expected  the following equality holds, because
$\BR_{1}(\sigma_0) = \BR_{1}^{0}(\sigma_0)$ and
$\BR_{1}^{\epsilon_1}(\sigma_0)\subseteq
\BR_{1}^{\epsilon_2}(\sigma_0)$ for all $\epsilon_1\leq \epsilon_2$, $\epsilon$ should be taken equal to $0$:
$$\ASV(\sigma_0)(v)=\sup_{\epsilon \geq 0 ~\mid~ \BR_{1}^{\epsilon}(\sigma_0)\neq\varnothing} ~~~\inf_{\sigma_1 \in \BR_{1}^{\epsilon}(\sigma_0)} \MPinf_0(\outcome_v(\sigma_0,\sigma_1))=\inf_{\sigma_1 \in \BR_{1}(\sigma_0)} \MPinf_0(\outcome_v(\sigma_0,\sigma_1))$$
\noindent
Finally, we note that changing the $\sup$ over $\epsilon$ into an $\inf$ in our definition, we get the classical notion of worst-case value in which the rationality of Player~1 and his payoff are ignored. We also  recall the definition of $\CSV$, the cooperative Stackelberg value:
$$\CSV(\sigma_0)(v)=\sup_{\epsilon \geq 0 ~\mid~ \BR_{1}^{\epsilon}(\sigma_0)\neq\varnothing} ~~~\sup_{\sigma_1 \in \BR_{1}^{\epsilon}(\sigma_0)} \MPinf_0(\outcome_v(\sigma_0,\sigma_1))
\mbox{~and~}\CSV(v)=\sup_{\sigma_0 \in \Sigma_0} \CSV(\sigma_0)(v)$$

\noindent

The interest reader is referred to~\cite{DBLP:conf/time/GuptaS14} for an in-depth treatment of this value. \color{black}{}

\subparagraph{The adversarial Stackelberg value may not be achievable} In contrast with results in~\cite{DBLP:conf/time/GuptaS14} that show that $\CSV$ can always be achieved, the following statement expresses the fact that the adversarial Stackelberg value may not be achievable but it can always be approximated by a strategy of Player~0.

\begin{figure}[t]
    \centering

\begin{tikzpicture}[->,>=stealth',shorten >=1pt,auto,node distance=2.8cm,
                    thick,scale=0.8,every node/.style={scale=0.9}]
  \tikzstyle{every state}=[fill=gray!30,text=black]
  \tikzstyle{every edge}=[draw=black]
  \tikzstyle{initial}=[initial by arrow]


	\node[thick,state,initial above,initial text=,rectangle,minimum size=25pt]
        (s0) at (0,0) {$v_0$};

	\node[thick,state,minimum size=25pt]
        (s1) at (-4,0) {$v_1$};

	\node[thick,state,minimum size=25pt]
        (s2) at (4,0) {$v_2$};

        \draw (s0) edge node {$(0,1)$} (s2) ;

        \draw (s2) edge [loop above] node {(0,1)}  (s2) ;

        \draw (s1) edge [loop above] node {(0,2)} (s1) ;

        \draw (s0) edge [bend left] node [above ]{(1,1)}  (s1) ;

        \draw (s1) edge [bend left] node [above] {(1,1)}  (s0) ;
\end{tikzpicture}

    \caption{In this game, $\ASV(v_0)=1$ but there is no Player~0 strategy to achieve this value.}
    \label{fig:no-opt-ASV}
\end{figure}
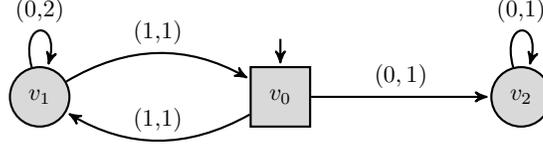

\begin{theorem}
\label{thm:equiv-def}
There exists a mean-payoff game $\mathcal{G}$ in which Player~0 has no strategy which enforces the adversarial Stackelberg value. Furthermore, for all mean-payoff games $\mathcal{G}$, for all vertices $v \in V$, for all $\epsilon >0$, there exists a strategy $\sigma_0 \in \Sigma_0$ such that $\ASV(\sigma_0)(v) > \ASV(v) - \epsilon$. Memory is needed to achieve high $\ASV$.
\end{theorem}
\begin{proof}[Proof sketch - full proof in Appendix]
First, consider the game depicted in Fig~\ref{fig:no-opt-ASV}. In this game, $\ASV(v_0)=1$ and it is not achievable. Player~0 needs to ensure that Player~1 does not take the transition from $v_0$ to $v_2$ otherwise she gets a payoff of $0$. To ensure this, Player~0 needs to choose a strategy (that cycles within $\{v_0,v_1\}$) and that gives to Player~1 at least $1+\epsilon$ with $\epsilon>0$. Such strategies gives $1-\epsilon$ to Player~0, and the value $1$ cannot be reached.

Second, by definition of the $\ASV$, the value is obtained as the sup over all strategies of Player~0. As a consequence, $\epsilon$-optimal strategies (for $\epsilon >0)$ exist.
\end{proof}




\subparagraph{Witnesses for the $\ASV$}
Given a mean-payoff game $\mathcal{G}$, we associate with each vertex $v$, the following set of pairs of real numbers:
 $\Lambda(v)=\{(c,d) \in \mathbb{R}^2 \mid v \models \Estrat{1} \MPinf_0 \leq c \land \MPinf_1  \geq d \}$.
We say that $v$ is $(c,d)$-bad if $(c,d) \in \Lambda(v)$. Let
  $c'\in \mathbb{R}$. A play $\pi$
in $\mathcal{G}$ is called a $(c',d)$-witness of $\ASV(v)>c$ if
it starts from $v$, $(\MPinf_0(\pi),\MPinf_1(\pi))=(c',d)$, $c' > c$ and $\pi$ does not
contain any $(c,d)$-bad vertex. A play $\pi$ is called a witness of
$\ASV(v)>c$ if it is a $(c',d)$-witness of $\ASV(v)>c$ for some
$c',d$. The following theorem justifies the name witness. 
\begin{theorem}
\label{thm:witness-mp}
Let $\mathcal{G}$ be a mean-payoff game and $v$ be one of its
vertices. $\ASV(v) > c$ iff there exists a play $\pi$ in $\mathcal{G}$
such that $\pi$ is a witness of $\ASV(v) > c$.
\end{theorem}
\begin{proof}
\emph{From right to left.} Assume the existence of a $(c',d)$-witness $\pi$ and let us show that there exists a strategy $\sigma_0$ which forces $\ASV(\sigma_0)(v) > c$. We define $\sigma_0$ as follows:
  \begin{enumerate}
      \item for all histories $h \leq \pi$ such that $\last(h)$ belongs to Player~0, $\sigma_0(h)$ follows $\pi$.
      \item for all histories $h \not\leq \pi$ where there has been a deviation from $\pi$ by Player~1, we assume that Player~0 switches to a strategy that we call {\em punishing}. This strategy is defined as follows. In the subgame after history $h'$ where $h'$ is a first deviation by Player~1 from $\pi$, we know that Player~0 has a strategy to enforce the objective: $\MPinf_0 > c \lor \MPinf_1 < d$. This is true because $\pi$ does not cross any $(c,d)$-bad vertex. So, we know that $h' \nvDash \Estrat{1} \MPinf_0 \leq c \land \MPinf_1 \geq d$ which entails the previous statement by determinacy of $n$-dimension mean-payoff games~\cite{DBLP:journals/iandc/VelnerC0HRR15} (here $n=2$).
      \item for all other histories $h$, Player~0 can behave arbitrarily as those histories are never reached when Player~0 plays as defined in point $1$ and $2$ above.
  \end{enumerate}
\noindent
Let us now establish  that the strategy $\sigma_0$ satisfies
$\ASV(\sigma_0)(v) > c$. We have to show the existence
  of some $\epsilon\geq 0$ such that
  $\BR^{\epsilon}_1(\sigma_0)\neq\varnothing$ and for all $\sigma_1\in
  \BR^{\epsilon}_1(\sigma_0)$,
  $\MPinf_0(\outcome_v(\sigma_0,\sigma_1))>c$ holds. For that, we consider two subcases:
\begin{enumerate}
    \item $\sup_{\sigma_1} \MPinf_1(\outcome_{v}(\sigma_0,\sigma_1))=d=\MPinf_1(\pi)$. This means that any strategy $\sigma_1$ of Player~1 that follows $\pi$ is for $\epsilon=0$ a best-response to $\sigma_0$. Now let us consider any strategy $\sigma_1 \in \BR^0_1(\sigma_0)$. Clearly, $\pi'=\outcome_{v}(\sigma_0,\sigma_1)$ is such that $\MPinf_1(\pi') \geq d$. If $\pi'=\pi$, we have that $\MPinf_0(\pi')=c' > c$. If $\pi'\neq\pi$, then when $\pi'$ deviates from $\pi$, we know that $\sigma_0$ behaves as the punishing strategy and so we have that $\MPinf_0(\pi') > c \lor \MPinf_1(\pi') <d$. But as $\sigma_1 \in \BR_1^0(\sigma_0)$, we conclude that $\MPinf_1(\pi') \geq d$, and so in turn, we obtain that $\MPinf_0(\pi') > c$.
    
    \item $\sup_{\sigma_1} \MPinf_1(\outcome_{v}(\sigma_0,\sigma_1))=d' > d$. Let $\epsilon >0$ be such that $d'-\epsilon>d$. By Lemma~\ref{lem:emptiness-nonemptiness}, $\BR_1^{\epsilon}(\sigma_0) \neq \varnothing$. Let us now characterize the value that Player~0 receives against any strategy $\sigma_1 \in \BR_1^{\epsilon}(\sigma_0)$. First, if $\sigma_1$ follows $\pi$ then Player~0 receives $c' > c$. Second, if $\sigma_1$ deviates from $\pi$, Player~1 receives at least $d'-\epsilon>d$. But by definition of $\sigma_0$, we know that if the play deviates from $\pi$ then Player~0 applies her punishing strategy. Then we know that the outcome satisfies $\MPinf_0 > c \lor \MPinf_1 < d$. But as $d'-\epsilon > d$, we must conclude that the outcome $\pi'$ is such that $\MPinf_0(\pi') > c$.  
\end{enumerate}

\noindent\emph{From left to right}. Let $\sigma_0$ such
that $\ASV(\sigma_0)(v) > c$. Then by the equivalence shown in the proof of Theorem~\ref{thm:equiv-def}, we know that
\begin{equation}\label{eq:epsilon}
\exists \epsilon \geq 0 : \BR^{\epsilon}_1(\sigma_0) \neq \varnothing \land  \forall \sigma_1 \in \BR_1^{\epsilon}(\sigma_0): \outcome_v(\sigma_0,\sigma_1) > c   
\end{equation}
Let $\epsilon^*$ be a value for $\epsilon$ that makes eq.~(\ref{eq:epsilon}) true. Take any $\sigma_1 \in \BR^{\epsilon^*}_1(\sigma_0)$ and consider $\pi=\outcome_{v}(\sigma_0,\sigma_1)$.  We will show that $\pi$ is a witness for $\ASV(v)>c$.

We have that $\MPinf_0(\pi) > c$. Let $d_1=\MPinf_1(\pi)$ and consider any $\pi' \in \outcome_{v}(\sigma_0)$. Clearly if $\MPinf_1(\pi')\geq d_1$ then there exists $\sigma'_1 \in \BR^{\epsilon^*}_1(\sigma_0)$ such that $\pi'=\outcome_{v_0}(\sigma_0,\sigma'_1)$ and we conclude that $\MPinf_0(\pi') > c$. So all deviations of Player~1 w.r.t. $\pi$ against $\sigma_0$ are either giving him a $\MPinf_1$ which is less than $d_1$ or it gives to Player~0 a $\MPinf_0$ which is larger than $c$. So $\pi$ is a $(\MPinf_0(\pi),\MPinf_1(\pi))$-witness for $\ASV(v) > c$ as we have shown that $\pi$ never crosses an $(c,\MPinf_1(\pi))$-bad vertex, and we are done.
\end{proof}

\noindent
The following statement is a direct consequence of the proof of the previous theorem.

\begin{corollary}\label{coro:up}
If $\pi$ is a witness for $\ASV(v) > c$ then all $\pi'$ such that:
      $\pi'(0)=v$,
      the set of vertices visited along $\pi$ and $\pi'$ are the same, and 
      $\MPinf_0(\pi') \geq \MPinf_0(\pi)$ and $\MPinf_1(\pi') \geq \MPinf_1(\pi)$, 
are also witnesses for $\ASV(v) > c$.
\end{corollary}

\subparagraph{Small witnesses and {\sf NP} membership}
Here, we refine Theorem~\ref{thm:witness-mp} to establish membership of the threshold problem to {\sf NP}.

\begin{theorem}
\label{NP-membership-ASV-MP}
Given a mean-payoff game $\game$, a vertex $v$
and a rational value $c \in \mathbb{Q}$, it can be decided in nondeterministic polynomial time if $\ASV(v) > c$.
\end{theorem}

\noindent Proof of Thm.~\ref{NP-membership-ASV-MP} relies on the existence of small witnesses established in the following lemma: 

\begin{lemma}
\label{lem:smallwitness}
Given a mean-payoff game $\game$, a vertex $v$ and  $c \in \mathbb{Q}$, $\ASV(v) > c$ if and only if there exists an {\sf SCC} reachable from $v$ that contains two simple cycles $\ell_1$, $\ell_2$ such that:
      $(i)$ there exist $\alpha,\beta \in \mathbb{Q}$ such that 
            $\alpha \cdot w_0(\ell_1) + \beta \cdot w_0(\ell_2) =c' > c$, and 
            $\alpha \cdot w_1(\ell_1) + \beta \cdot w_1(\ell_2)=d$
      $(ii)$ there is no $(c,d)$-bad vertex $v'$ along the path from $v$ to $\ell_1$, the path from $\ell_1$ to $\ell_2$, and the path from $\ell_2$ to $\ell_1$.
\end{lemma}
\begin{proof}[Proof sketch - full proof in Appendix]
Theorem~\ref{thm:witness-mp} establishes the existence of a witness
$\pi$ for $\ASV(v)>c$. In turn, we show here that the existence of
such a $\pi$ can be established by a polynomially checkable witness
composed of the following elements. First, a simple path from $v$ to
the {\sf SCC} in which $\pi$ gets trapped in the long run, $(ii)$ two
simple cycles (that can produce the value $(c',d)$ of $\pi$) by
looping at the right frequencies along the two cycles. Indeed,
$(\MPinf_0(\pi),\MPinf_1(\pi))$ only depends on the suffix in the {\sf
  SCC} in which it gets trapped. Furthermore, by
Theorem~\ref{thm:caractSCC}, Proposition~\ref{prop:larger} and
Corollary~\ref{coro:up}, we know that the mean-payoff of witnesses can
be obtained as the convex combination of the mean-payoff coordinates
of simple cycles, and 3 such simple cycles are sufficient by the
CarathÃ©odory baricenter theorem. A finer analysis of the geometry of
the sets allows us to go to 2 cycles only (see the full proof in
appendix). 
\end{proof}

\begin{proof}[Proof of Theorem~\ref{NP-membership-ASV-MP}]
According to Lemma~\ref{lem:smallwitness}, the nondeterministic algorithm that establishes the membership to {\sf NP}  guesses a reachable {\sf SCC} together with the two simple cycles $\ell_1$ and $\ell_2$, and parameters $\alpha$ and $\beta$. Additionally, for each vertex $v'$ that appears along the paths to reach the {\sf SCC}, on the simple cycles $\ell_1$ and $\ell_2$, and to connect those simple cycles, the algorithm guesses a memoryless strategy $\sigma^{v'}_0$ for Player~0 that establishes $v' \nvDash \Estrat{1} \MPinf_0 \leq c \land \MPinf_1 \geq d$ which means by determinacy of multi-dimensional mean-payoff games, that $v' \vDash \Estrat{0} \MPinf_0 > c \lor \MPinf_1 < d$. The existence of those memoryless strategy is established in Propositions~\ref{prop:mp-inverse} and~\ref{prop:mp-inf-sup} given in appendix (in turn those propositions rely on results from~\cite{DBLP:journals/iandc/VelnerC0HRR15}). Those memoryless strategies are checkable
in \textsf{PTime}~\cite{DBLP:journals/dm/Karp78}.
\end{proof}

\subparagraph{Computing the $\ASV$ in mean-payoff games}
The previous theorems establish the existence of a notion of witness for the adversarial Stackelberg value in non zero-sum two-player mean-payoff games. This notion of witness can be used to decide the threshold problem in \textsc{NPtime}. We now show how to use this notion to effectively compute the $\ASV$. This algorithm is also based on the computation of an effective representation, for each vertex $v$ of the game graph, of the infinite set of pairs $\Lambda(v)$. The following lemma expresses that a symbolic representation of this set of pairs can be constructed effectively. This result is using techniques that have been introduced in~\cite{DBLP:conf/cav/BrenguierR15}.


\begin{lemma}\label{lem:badval-effective}
Given a bi-weighted game graph $\game$ and a vertex $v \in V$, we
can effectively construct a formula $\Psi_v(x,y)$ of $\langle \mathbb{R},+,\leq \rangle$ with two free variables such that 
$(c,d) \in \Lambda(v)$ if and only if the formula $\Psi_v(x,y)[x/c,y/d]$ is true.
\end{lemma}

\subparagraph{Extended graph game} From the graph game
$\game=(V,E,w_0,w_1)$, we construct the extended graph game
$\game^{{\sf ext}}=(V^{\sf ext},E^{\sf ext},w_0^{\sf ext},w_1^{\sf
  ext})$, whose vertices and edges are defined as follows. The set of
vertices is $V^{\sf ext}=V \times 2^V$. With an history $h$ in $\game$, we associate a vertex in $\game^{\sf ext}$ which is a pair $(v,P)$, where $v={\last}(h)$ and $P$ is the set of the vertices traversed along $h$. Accordingly the set of edges and the weight functions are defined as
$E^{\sf ext}=\left\{ ((v,P),(v',P')) \mid (v,v') \in E \land P'=P\cup\{v'\}\right\}$ and
$w_i^{\sf ext}((v,P),(v',P'))=w_i((v,v'))$, for $i \in \{0,1\}$. Clearly, there exists a bijection between the plays $\pi$ in $\game$ and the plays $\pi^{\sf ext}$ in $\game^{\sf ext}$ which start in vertices of the form $(v,\{v\})$, i.e. $\pi^{\sf ext}$ is mapped to the play $\pi$ in $\game$ that is obtained by erasing the second dimension of its vertices. 

\begin{proposition}
\label{prop:extended-plays} For all game graph $\game$, the following holds:
  \begin{enumerate}
         
      \item Let $\pi^{\sf ext}$ be an infinite play in the extended graph and
$\pi$ be its projection into the original graph $\game$ (over the first component of each vertex) , the following properties hold: 
      $(i)$ For all $i < j$: if $\pi^{\sf ext}(i)=(v_i,P_i)$ and
        $\pi^{\sf ext}(j)=(v_j,P_j)$ then $P_i \subseteq P_j$.
      $(ii)$ $\MPinf_i(\pi^{\sf ext})=\MPinf_i(\pi)$, for $i \in \{0,1\}$.
   \item 
   The unfolding of $\game$ from $v$ and the unfolding of $\game^{\sf ext}$ from $(v,\{v\})$ are isomorphic, and so $\ASV(v)=\ASV(v,\{v\})$.
  \end{enumerate}

\end{proposition}

By the first point of the latter proposition and since the set of vertices of the graph is finite, the second component of any play $\pi^{\sf ext}$ stabilises into a set of vertices of $\game$ which we denote by $V^*(\pi^{\sf ext})$. 
  


We now show how to characterize $\ASV(v)$ with the notion of witness introduced above and the decomposition of $G^{\sf ext}$ into {\sf SCC}. This is formalized in the following lemma:

\begin{lemma}
\label{lem:decomp}
For all mean-payoff games $\mathcal{G}$, for all vertices $v \in V$,
let ${\sf SCC}^{\sf ext}(v)$ be the set of strongly-connected components in $\game^{\sf ext}$ which are reachable from $(v,\{v\})$, then we have 
$$\ASV(v)=\max_{S\in {\sf SCC}^{\sf ext}(v)} \sup \{ c \in \mathbb{R} \mid \exists \pi^{\sf ext}: \pi^{\sf ext}~\mbox{is a witness for}~\ASV(v,\{v\})>c~\mbox{and}~V^*(\pi^{\sf ext})=S \}$$
\end{lemma}
\begin{proof}
First, we note the following sequence of equalities:
  $$
  \begin{array}{l}
    \ASV(v)\\=\sup \{ c \in \mathbb{R} \mid \ASV(v) \geq c \} \\
            =\sup \{ c \in \mathbb{R} \mid \ASV(v) > c \} \\
            =\sup \{ c \in \mathbb{R} \mid \exists \pi: \pi~\mbox{is a witness for}~\ASV(v)>c \} \\
            =\sup \{ c \in \mathbb{R} \mid \exists \pi^{\sf ext}: \pi^{\sf ext}~\mbox{is a witness for}~\ASV(v,\{v\})>c \} \\
            =\max_{S\in {\sf SCC}^{\sf ext}(v)} \sup \{ c \in \mathbb{R} \mid \exists \pi^{\sf ext}: \pi^{\sf ext}~\mbox{is a witness for}~\ASV(v,\{v\})>c~\mbox{and}~V^*(\pi^{\sf ext})=S \} \\
  \end{array}
  $$
  \noindent
  The first two equalities are direct consequences of the definition of the supremum and that $\ASV(v) \in \mathbb{R}$. The third is a consequence of Theorem~\ref{thm:witness-mp} that guarantees the existence of witnesses for strict inequalities.  The fourth equality is a consequence of point 2 in Proposition~\ref{prop:extended-plays}. The last equality is the consequence of point 1 in Proposition~\ref{prop:extended-plays}. 
\end{proof}

By definition of $\game^{\sf ext}$, for all {\sf SCC} $S$ of $\game^{\sf
  ext}$, there exists a set of vertices of $\game$ which we also
denote by $V^*(S)$ such that any vertex of $S$ is of the form
$(v,V^*(S))$. The set of bad thresholds for $S$ is then defined as 
$\Lambda^{\sf ext}(S) = \bigcup_{v\in V^{*}(S)} \Lambda(v)$. 
Applying Lemma~\ref{lem:badval-effective}, we can construct a formula $\Psi_S(x,y)$ which symbolic encodes the set $\Lambda^{\sf ext}(S)$.

Now, we are equipped to prove that $\ASV(v)$ is effectively computable. This is expressed by the following theorem and established in its proof.

\begin{theorem}
\label{thm:effectiveASV}
For all mean-payoff games $\mathcal{G}$, for all vertices $v \in V$, the value $\ASV(v)$ can be effectively expressed by a formula $\rho_{v}$ in $\langle \mathbb{R},+,\leq \rangle$ and explicitly computed from this formula.
\end{theorem}
\begin{proof}
To establish this theorem, we show how to build the formula $\rho_{v}(z)$ that is true iff $\ASV(v)=z$. We use Lemma~\ref{lem:decomp}, to reduce this to the construction of a formula that expresses the existence of witnesses for $\ASV(v)$ from $(v,\{v\})$:
$$\ASV(v)=\!\!\!\!\!\max_{S\in {\sf SCC}^{\sf ext}(v)}\!\! \sup \{ c \in \mathbb{R}
\mid \exists \pi^{\sf ext}: \pi^{\sf ext}~\text{is a witness for}~\ASV(v,\{v\})>c~\text{and}~V^*(\pi^{\sf ext})=S \} $$
\noindent
As $\max_{S\in {\sf SCC}^{\sf ext}(v)}$ is easily expressed in $\langle \mathbb{R},+,\leq \rangle$, we concentrate on one {\sf SCC} $S$ reachable from $(v,\{v\})$ and we show how to express 
$$\sup \{ c \in \mathbb{R} \mid \exists \pi^{\sf ext}: \pi^{\sf ext}~\mbox{is a witness for}~\ASV(v,\{v\})>c~\mbox{and}~V^*(\pi^{\sf ext})=S \}$$
\noindent
First, we define a formula that express the existence of a witness for $\ASV(v)>c$.
This is done by the following formula:
$$\rho^S_{v_0}(c) \equiv \exists x,y \cdot x > c \land \Phi_S(x,y) \land \neg \Psi_S(c,y)$$

\noindent
Where $\Phi_S(x,y)$ is the symbolic encoding of $F_{\sf min}({\sf CH}(\mathbb{C}(S)))$ as defined in Theorem~\ref{thm:caractSCC}. This ensures that the values $(x,y)$ are the mean-payoff values realizable by some path in $S$. By Lemma~\ref{lem:badval-effective}, $\neg \Psi_S(c,y)$ expresses that the path does not cross a $(c,y)$-bad vertex. So the conjunction $\exists x,y \cdot x > c \land \Phi_S(x,y) \land \neg \Psi_S(c,y)$ establishes the existence of a witness with mean-payoff values $(x,y)$ for the threshold $c$. 
From this formula, we can compute the $\ASV$ by quantifier elimination in:
$$\exists z \cdot \forall e>0 \cdot  (\rho^S_{v_0}(z-e) \land (\forall y \cdot \rho^S_{v_0}(y) \implies y \leq z))$$
\noindent
and obtain the unique value of $z$ that makes the formula true.
\end{proof}

\section{Stackelberg values for discounted-sum games}\label{secDS}
In this section, we study the notion of Stackelberg value in the case of discounted sum measures. Beside the adversarial setting considered so far, we also refer to a \emph{cooperative} framework for discounted sum-games, since we add some results to~\cite{DBLP:journals/corr/GuptaS14c}, where the cooperative Stackelberg value for discounted-sum measures has been previously introduced and studied.


\subparagraph{On the existence of best-responses for $\sf{DS}$}
First, we show that the set of best-responses for Player $1$ to strategies of Player $0$ is guaranteed to be nonempty for discounted sum games, while this was not the case in mean-payoff games. 

  
\begin{lemma}
\label{lem:exists-ds-br}
For all discounted sum games  $\mathcal{G}$ and strategies $\sigma_0\in \Sigma_0(\mathcal{G})$,   $\BR_1(\sigma_0)\neq \varnothing$.
\end{lemma}
\begin{proof}
Given $\sigma\in\Sigma_0(\mathcal{G})$, consider $S=\{{\sf
  DS}_1(\out(\sigma,\tau))\:|\:\tau\in
\Sigma_1(\mathcal{G})\}$. $S$ is a non empty limited subset of
$\mathbb{R}$, since for each $\tau\in\Sigma_1(\mathcal{G})$
 it holds  ${{\sf{DS}}_1(\out(\sigma,\tau))}\leq \dfrac{W}{1-\lambda}$, where $W$ is the maximum absolute value of a weight in $\mathcal{G}$. Hence, $S$ admits a unique superior extreme $s=sup(S)$. By definition of superior extreme, for each $\epsilon >0$, there exists $v_\epsilon\in S$ such that $s\geq v_\epsilon>s-\epsilon$. Therefore, for each  $\epsilon>0$ there exists $\tau_\epsilon\in\Sigma_1(\mathcal{G})$ such that $s\geq {\sf{DS}_1}(\out(\sigma,\tau_\epsilon))> s -\epsilon$, i.e.:
 \begin{equation}\label{equaTowardCont}
 0\leq s- {\sf{DS}}_1(\out(\sigma,\tau_\epsilon)) < \epsilon
 \end{equation}
 We show that this implies that $\out(\sigma)$ contains a play $\pi^*$ such that ${\sf DS}_1(\pi^*)=s$, which leads to $BR_1(\sigma)\neq \varnothing$, since Player $1$ has a strategy to achieve $s=sup(\{{\sf DS}_1(\out(\sigma,\tau))\:|\:\tau\in \Sigma_1(\mathcal{G})\})$.

\noindent By contradiction, suppose that ${\out}(\sigma)$ does not contain any play $\pi$ such that ${\sf DS}_1(\pi)=s$. Hence, for each $\pi\in {\out}(\sigma)$, it holds that ${\sf DS}_1(\pi)<s$ and $\pi$ admits a prefix $\pi_{\leq k}$ such that: 
\begin{equation}\label{equa2}
{\sf DS}_1(\pi_{\leq k})+ W\frac{\lambda^k}{1-\lambda}<s
\end{equation}
Hence,  we can cut each play in $\out(\sigma)$ as soon as Equation \ref{equa2} is accomplished, leading to a finite tree $T$ (by Konig lemma, since $\out(\sigma)$ is finitely branching). Let $\pi^*_T=\pi_0\dots \pi_k$ be a branch in the finite tree $T$  such that  the value $v(\pi^*_T)= s -({\sf DS}_1(\pi^*_T)+ W\frac{\lambda^k}{1-\lambda})$ is minimal. Note that, by Equation \ref{equa2},  $v(\pi^*_T)>0$ since $v(\pi^*_T)= s -({\sf DS}_1(\pi^*_T)+ W\frac{\lambda^k}{1-\lambda})>s-s=0$. 

Then, for  each play $\pi$, let $\pi_{\leq p}$ be the longest prefix of $\pi$ which is also a branch in the finite tree $T$. By definition of   $\pi^*_T$, we have:
\begin{equation}\label{eq4}
  s-{\sf DS}_1(\pi)\geq s -({\sf DS}_1(\pi_{\leq p})+ W\frac{\lambda^p}{1-\lambda})\geq v(\pi^*_T) >0
 \end{equation}
 This leads to  a contradiction to the fact that for all $\epsilon>0$ there exists $\tau\in\Sigma_1(\mathcal{G})$ such that $ s -{\sf DS}_1(\out(\sigma,\tau_\epsilon))<\epsilon $, established within Equation \ref{equaTowardCont}.
\end{proof}

\subparagraph{Stackelberg values for $\sf{DS}$ in the adversarial and cooperative settings}

The existence of best-responses allows us to simplify the notion of Stackelberg value for discounted sum measures, avoiding the parameter $\epsilon$ used for mean-payoff games. In particular, 
the adversarial Stackelberg value {\sf ASV}$(v)$  for discounted sum games is defined for all $\sigma_0\in\Sigma_0(\mathcal{G})$ as:
   $$\ASV(\sigma_0)(v)=\inf_{\sigma_1 \in \BR_{1}(\sigma_0)} {\sf DS}_0^\lambda(\outcome_v(\sigma_0,\sigma_1))\mbox{~and ~} \ASV(v)=\sup_{\sigma_0\in\Sigma_0} \ASV(\sigma_0)(v)$$
\noindent As previously announced, we also consider the notion of Stackelberg value for discounted sum measures in the cooperative setting, 
where Player $0$ suggests a profile of strategies $(\sigma_0,\sigma_1)$ and Player $1$ agrees to play $\sigma_1$ if the latter strategy is a best response to $\sigma_0$. Formally, the cooperative  Stackelberg value {\sf CSV}$(v)$  for discounted sum games is defined as:
$$\CSV(\sigma_0)(v)=\sup_{\sigma_1 \in \BR_{1}(\sigma_0)} {\sf DS}_0^\lambda(\outcome_v(\sigma_0,\sigma_1))\mbox{~and ~} \CSV(v)=\sup_{\sigma_0\in\Sigma_0} \CSV(\sigma_0)(v)$$

\noindent Lemma \ref{lem:target-csv} below links the cooperative Stackelberg value for discounted-sum measures to the \emph{target discounted-sum problem}~\cite{DBLP:conf/lics/BokerHO15} (cfr. Definition \ref{tdsp}), whose decidability is notoriously hard to solve and relates to  several open 
questions in mathematics and computer science \cite{DBLP:conf/lics/BokerHO15}.

\begin{definition}[Target Discount Sum Problem \cite{DBLP:conf/lics/BokerHO15} ({\sf TDS})] \label{tdsp} Given a rational discount factor $0<\lambda<1$ and three rationals $a,b,t$ does there exist an infinite sequence $w\in\{a,b\}^\omega$ such that $\sum_{i=0}^\infty w(i)\lambda^i=t$?
 \end{definition}
\noindent In particular, given an instance $I=(a,b,t,\lambda)$ of the
{\sf TDS} problem, Figure \ref{figTDSreduction}  depicts  a discounted sum game  $\mathcal{G}^I$ such that $I$ admits a solution iff $\CSV(v)\geq \lambda\cdot t$.

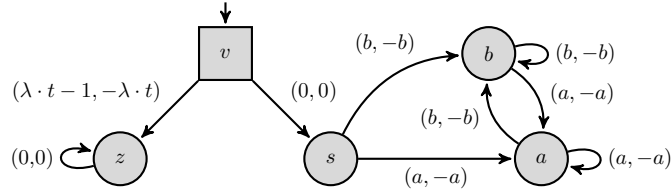
\begin{figure}[t]
    \centering

\begin{tikzpicture}[->,>=stealth',shorten >=1pt,auto,node distance=2.8cm,
                    thick,scale=0.7,every node/.style={scale=0.8}]
  \tikzstyle{every state}=[fill=gray!30,text=black]
  \tikzstyle{every edge}=[draw=black]
  \tikzstyle{initial}=[initial by arrow]


	\node[thick,state,initial above,initial text=,rectangle,minimum size=25pt]
        (s1) at (0,0) {$v$};

	\node[thick,state,minimum size=25pt]
        (s2) at (-2,-2) {$z$};
    \node[thick,state,minimum size=25pt]
        (s3) at (2,-2) {$s$};  
    \node[thick,state,minimum size=25pt]
        (s4) at (6,-2) {$a$};
    \node[thick,state,minimum size=25pt]
        (s5) at (5,0) {$b$};     


       \draw (s2) edge [loop left] node {(0,0)} node [left,xshift=-0.1cm,yshift=-0.35cm] {} (s2) ;

       \draw (s1) edge [swap] node {$(\lambda\cdot t-1, -\lambda \cdot t)$} node [below] {} (s2) ;
       
         \draw (s1) edge  node {$(0,0)$} node [below] {} (s3) ;

        \draw (s4) edge [bend left] node [near end]{$(b,-b)$} node [left,xshift=-0.1cm,yshift=-0.35cm] {} (s5) ;

 \draw (s5) edge [ bend left] node [near end]{$(a,-a)$} node [left,xshift=-0.1cm,yshift=-0.35cm] {} (s4) ;
 
 \draw (s3) edge [bend left] node [near end]{$(b,-b)$} node [left,xshift=-0.1cm,yshift=-0.35cm] {} (s5) ;
 
 \draw (s3) edge  node [swap]{$(a,-a)$} node [left,xshift=-0.1cm,yshift=-0.35cm] {} (s4) ; 
 
  \draw (s4) edge [loop right] node {$(a,-a)$} node [left,xshift=-0.1cm,yshift=-0.35cm] {} (s4) ;
  \draw (s5) edge [loop right] node {$(b,-b)$} node [left,xshift=-0.1cm,yshift=-0.35cm] {} (s5) ;



       
\end{tikzpicture}

    \caption{The instance of ${\sf TDS}$ $I=(a,b,\lambda,t)$ admits a solution iff $\CSV(v)\geq \lambda\cdot t$.}
    \label{figTDSreduction}
\end{figure}

\begin{lemma}
\label{lem:target-csv}
The  target discounted-sum problem reduces to the  problem of deciding if $\CSV(v)\geq c$ in discounted-sum games.
\end{lemma}
\begin{proof}
Let $I=(a,b,t,\lambda)$ be an instance of the target discounted sum problem and consider the game $\mathcal{G}^I$ depicted in Figure \ref{figTDSreduction}. We prove that $I$ admits a solution iff $\CSV(v)\geq \lambda\cdot t$.

Suppose that $I$ admits a solution and let $w\in\{a,b\}^\omega$ such that $\sum_{i=0}^\infty w(i)\lambda^i=t$. Consider the following strategy $\sigma$ for Player $0$: for all $\alpha\in \{a,b\}^*$, $\sigma(vs\alpha)=x$ if $w(|\alpha|)=x$, where $x\in \{a,b\}$. We prove that if $\tau$ is a best response to $\sigma$, then ${\DS_0(\out(\sigma,\tau))}=\lambda\cdot t$. In fact, Player $1$ has two choices from $v$. Let us denote $\tau_s$ (resp. $\tau_z$) the strategy that prescribes to Player $1$ to proceed to vertex $s$ (resp. $z$) out from $v$. We have that ${\DS_1(\out(\sigma,\tau_s))}={\DS_1(\out(\sigma,\tau_z))}=-\lambda\cdot t$, by definition of $\sigma$ and $\mathcal{G}^I$. Hence, $\tau_s$ is a best response to $\sigma$ which guarantees to Player $0$ a payoff ${\DS_0(\out(\sigma,\tau_s))}=\lambda\cdot t$.

In the other direction, suppose that $I$ does not admit any solution, i.e. there does not exist an infinite sequence $w\in\{a,b\}^\omega$ such that $\sum_{i=0}^\infty w(i)\lambda^i=t$. We prove that for any strategy $\sigma$ for Player $0$, if $\tau$ is a best response of Player $1$ to $\sigma$ then 
${\DS_0(\out(\sigma,\tau))}<\lambda\cdot t$. Let $\sigma$ be an arbitrary strategy for Player $0$ and consider the strategy $\tau_z$ for Player $1$.

We have two cases to consider depending on wether $\tau_z$ is a best response to $\sigma$ or not.
In the first case, we have that ${\DS(\out(\sigma,\tau_z))}=(\lambda\cdot t-1,-\lambda\cdot t)$ and, since $\tau_z$ is a best response to $\sigma$, we need to have ${\DS_1(\out(\sigma,\tau_s))}\leq -\lambda\cdot t$. We can not have that ${\DS_1(\out(\sigma,\tau_s))}= -\lambda\cdot t$, since this would imply ${\DS_0(\out(\sigma,\tau_i))}= {-\DS_1(\out(\sigma,\tau_s))}= \lambda\cdot t$ contradicting our hypothesis that $I$ does not admit any solution. Therefore, ${\DS_1(\out(\sigma,\tau_s))} < -\lambda\cdot t$, meaning that $\tau_s$ is not a best response to $\sigma$ and ${\CSV(v)}=\lambda\cdot t-1<\lambda\cdot t$. 

In the second case, where $\sigma_z$ is not a best response to $\sigma$, we have that ${\DS_1(\out(\sigma,\tau_s))} > -\lambda\cdot t$ which implies that ${\CSV(v)}={\DS_0(\out(\sigma,\tau_s))} = {-\DS_1(\out(\sigma,\tau_s))} < \lambda\cdot t$.  
\end{proof}

The construction used to link the cooperative Stackelberg value to the target discounted sum problem can be properly modified\footnote{Consider the game $\mathcal{G}^I$ depicted in Figure \ref{figTDSreduction}  for $a=0,b=1,\lambda=\frac{2}{3}, t=\frac{3}{2}$. By Proposition $1$ in \cite{ChattarjeeLPAR2013}, Player $0$ can achieve $\frac{3}{2}$ from $s$---and therefore \CSV(v)=1---only with infinite memory.} to prove that infinite memory may be necessary to allow Player~0 to achieve her $\CSV$, recovering a result originally proved in~\cite{DBLP:journals/corr/GuptaS14c}. In the same paper, the authors show that in $3$-player discounted sum games the cooperative Stackelberg value cannot be approximated by considering strategies with bounded memory only. In the next section, we show that this is not the case  for $2$-player discounted sum games.

\subparagraph{{\bf Gap problems and their algorithmic solutions}}
We consider a gap approximation of the Stackelberg value problem in
both the cooperative and the adversarial settings. Given $\epsilon>0$
and $c\in \mathbb{Q}$, and ${\sf VAL}\in\{\CSV,\ASV\}$, let us define the sets of games:
\begin{itemize}
    \item ${\sf Yes}^{\epsilon,c}_{\sf VAL} =\{(\mathcal{G},v)\:|\:
      \mathcal{G} \mbox{ is a game with } {\sf VAL}(v) >c+\epsilon\}$
    \item ${\sf No}^{\epsilon,c}_{\sf VAL} =\{(\mathcal{G},v)\:|\:
      \mathcal{G} \mbox{ is a game with } {\sf VAL}(v) <c-\epsilon\}$
\end{itemize}
The \CSV-gap (resp. \ASV-gap)  problem with gap $\epsilon>0$ and threshold $c\in \mathbb{Q}$ consists in determining if  a given game $\mathcal{G}$ belongs to ${\sf Yes}^{\epsilon,c}_\CSV$ or ${\sf No}^{\epsilon,c}_\CSV$ (resp. ${\sf Yes}^{\epsilon,c}_\ASV$ or ${\sf No}^{\epsilon,c}_\ASV$). More precisely, solving the Stackelberg value gap  problem in e.g. the cooperative setting  amounts to answer {\sf Yes} if the instance of the game belongs to ${\sf Yes}^{\epsilon,c}_\CSV$, answer {\sf No} if the instance belongs to ${\sf No}^{\epsilon,c}_\CSV$, never answer or answer arbitrarily otherwise.

Theorem \ref{thm:ASV-CSV-gap} below uses the results in Lemma \ref{lemmaApprox} to provide an algorithm that solves the Stackelberg value gap problem in the cooperative and adversarial settings,  for games with discounted sum objectives. In particular, Lemma \ref{lemmaApprox}, shows   that finite memory strategies are sufficient to witness  Stackelberg values strictly greater than a threshold $c\in \mathbb{Q}$.

\begin{lemma}\label{lemmaApprox}
Let $\mathcal{G}$ be a discounted-sum game and  consider
$c\in\mathbb{Q}$ and $\epsilon>0$. If Player $0$ has a strategy
$\sigma_0$ such that $\CSV(\sigma_0)(v)>c+\epsilon$
(resp. $\ASV(\sigma_0)(v)>c+\epsilon$), then Player $0$ has a strategy
$\sigma_0^*$ with finite memory $M(\epsilon)$ such that
$\CSV(\sigma_0^*)(v) > c$ (resp. $\ASV(\sigma_0^*)(v) > c$). Moreover,
$M(\epsilon)$ is computable given $\epsilon$. 
\end{lemma}
\begin{proof} 
Let $\sigma^{{\sf DS}_1}_{min}\in \Sigma_0$  be a  memoryless strategy for Player $0$ minimizing $sup_{\tau\in\Sigma_1} {\sf DS}_1 (\out(\sigma,\tau))$.   Let $\sigma^{{\sf DS}_1}_{max}\in \Sigma_0$  be a  memoryless strategy for Player $0$ that maximizes $sup_{\tau\in\Sigma_1} {\sf DS}_1 (\out(\sigma,\tau))$. Such  memoryless strategies exist since $2$-player (single-valued) discounted-sum games are memoryless determined. In particular,  $\sigma^{{\sf DS}_1}_{min}\in \Sigma_0$  can be obtained by using standard algorithms for two players (single-valued) discounted-sum games. In turn, $\sigma^{{\sf DS}_1}_{max}\in \Sigma_0$  can be computed by solving a single player (single valued) discounted-sum game, in which all the nodes are controlled by the maximizer who aims at maximizing ${\sf DS}_1$.

\emph{Cooperative Setting}: Let  $\sigma^*\in\Sigma_0(\mathcal{G})$ be a strategy for Player $0$ such that $\sf{DS}_0( {\out}(\sigma^*,\tau))> c+\epsilon$ for some strategy  $\tau\in BR_1(\sigma^*)$.  Denote by $\pi^*$ the play $\pi^*=\out(\sigma^*,\tau)$ and let $N$ such that $\lambda^N\dfrac{W}{1-\lambda}<\dfrac{\epsilon}{2}$. Given the above premises, consider the finite memory strategy $\sigma'\in\Sigma_0$ for Player $0$ that 
follows $\sigma^*$ for the first $N$ steps and then either apply the memoryless strategy $\sigma^{{\sf DS}_1}_{min}\in \Sigma_0$ or the memoryless strategy $\sigma^{{\sf DS}_1}_{max}\in \Sigma_0$, depending on the history $h$ followed up to $N$. In particular, if $h=\pi^*_{\leq N}$, then the strategy $\sigma'$ prescribes to Player $0$ to follow $\sigma^{{\sf DS}_1}_{max}\in \Sigma_0$, cooperating with Player $1$ at maximizing ${\sf DS}_1$. Otherwise ($h\neq\pi^*_{\leq N}$), the strategy $\sigma'$ prescribes to Player $0$ to follow $\sigma^{{\sf DS}_1}_{min}\in \Sigma_0$, minimizing the payoff of the adversary. We show that a best response $\tau'$ for Player $1$ to $\sigma'$ consists in following $\pi^*$ up to $N$ and then applying the memoryless strategy $\tau^{{\sf DS}_1}_{max}\in \Sigma_1$, i.e.  maximizing $sup_{\sigma\in\Sigma_0} {\sf DS}_1 (\out(\sigma,\tau))$. Infact, by definition of $\sigma'$ and $\tau'$ we have that:
\begin{itemize}
    \item ${\sf \DS}_1(\out(\sigma',\tau'))\geq {\sf \DS}_1(\pi^*)$
     \item for any other strategy $\tau''\neq \tau'$ for Player $1$: 
    \begin{itemize}
    \item if ${\out}(\sigma',\tau'')_{\leq N}=x\neq \pi^*_{\leq N}$, then:
    $${\sf DS}_1({\out}(\sigma',\tau''))={\sf DS}_1(x) + \lambda^N{\sf DS}_1({\out}_x(\sigma^{{\sf DS}_1}_{min},\tau''))\leq$$
     $$\leq {\sf DS}_1(x) + \lambda^N(sup_{\tau\in \Sigma_1}({\sf DS}_1({\out}_x(\sigma^{{\sf DS}_1}_{min},\tau)))\leq$$
      $$\leq {\sf DS}_1(x) + \lambda^N(sup_{\tau\in \Sigma_1}({\sf DS}_1({\out}_x(\sigma^*,\tau)))= {\sf DS}_1(\pi^*)\leq {\sf \DS}_1({\out}(\sigma',\tau'))$$
      since ${\sf DS}_1(\pi^*)$ is the payoff (for player $1$) of a best response of Player $1$ to $\sigma^*$.
    \item  if ${\out}(\sigma',\tau'')_{\leq N}=x= \pi^*_{\leq N}$, then:
    $${\sf DS}_1({\out}(\sigma',\tau''))\leq {\sf DS}_1(x)+ \lambda^N\cdot sup\{{\sf DS}_1(\pi)\:|\:\pi\in {\sf Plays}(\mathcal{G}) \wedge \pi \mbox{ starts at } last(x)\}=$$
    $$= {\sf DS}_1(x)+ \lambda^N\cdot {\sf DS}_1({\out}_x(\sigma^{{\sf DS}_1}_{max},\tau^{{\sf DS}_1}_{max}))={\sf DS}_1({\out}(\sigma',\tau'))$$
    \end{itemize}

\end{itemize}

Finally, we show that the best response $\pi'$ of Player $1$ to $\sigma'$ guarantees to Player $0$ a payoff greater than $c$. Infact, 
  ${\sf{DS}}_0({\out}(\sigma',\tau'))>{\sf{DS}}_0(\pi^*_{\leq N})-\dfrac{\epsilon}{2}>c+\dfrac{\epsilon}{2}-\dfrac{\epsilon}{2}=c$, since ${\sf{DS}}_0(\pi^*_{\leq N})>c+\dfrac{\epsilon}{2}$. Due to the choice of $N$, having ${\sf{DS}}_0(\pi^*_{\leq N})\leq c+\dfrac{\epsilon}{2}$ would lead infact to the following contradiction: ${\sf{DS}}_0(\pi^*)\leq {\sf{DS}}_0(\pi^*_{\leq N}) + \lambda^N\dfrac{W}{1-\lambda}<{\sf{DS}}_0(\pi^*_{\leq N})+\dfrac{\epsilon}{2} \leq c+\dfrac{\epsilon}{2} +\dfrac{\epsilon}{2} = c+\epsilon$, i.e. ${\sf{DS}}_0(\pi^*)\leq c+\epsilon$.

\emph{Adversarial Setting}: Let  $\sigma\in\Sigma_0$ be a strategy for Player $0$ such that for all $\tau\in BR_1(\sigma)$ it holds $\sf{DS}_0(\out(\sigma,\tau))> c+\epsilon$. Let $N$ such that $\lambda^N\dfrac{W}{1-\lambda}<\dfrac{\epsilon}{2}$ and consider the unfolding $T$ of $\out(\sigma)$ up to $N$. For each maximal root-to-leaf branch $b$ of $T$, color its leaf $last(b)$  green if $b$ is the prefix $\pi_{\leq N}$ of some play $\pi=\out(\sigma, \tau)$ such that $\tau\in BR_1(\sigma)$. Otherwise, let the leaf $last(b)$ of $b$ be colored by red. 
We show that the finite memory strategy $\sigma^* \in\Sigma_0$ that prescribes to Player $0$ to follow $\sigma$ up to $N$ and then:

\begin{itemize}
\item from each green node apply the  memoryless strategy $\sigma^{{\sf DS}_1}_{max}\in \Sigma_0$ (i.e. cooperate with Player $1$ to maximize the payoff $\sf{DS}_1$)
\item from each red node apply the memoryless strategy $\sigma^{{\sf DS}_1}_{min}\in \Sigma_0$ (i.e. minimize  the payoff $\sf{DS}_1$ of the adversary )
\end{itemize}
is such that $ASV(\sigma^*)>c$.
Let $d=sup\{{\sf{DS}}_1(\out(\sigma,\tau))\:|\:
\tau\in \Sigma_1(\mathcal{G})\}$ and consider $\pi\in \out(\sigma^*)$. 

First, we show that if $\pi$ contains a green node then 
$\sf{DS}_0(\pi_{\leq N})>c$. In fact, ${\sf{DS}}_0(\pi_{\leq N})>{\sf{DS}}_0(\pi_{\leq N}) - \lambda^N\dfrac{W}{1-\lambda}> c+\dfrac{\epsilon}{2}-\dfrac{\epsilon}{2}=c$, since 
$\lambda^N\dfrac{W}{1-\lambda}<\dfrac{\epsilon}{2}$ by definition of $N$ and since ${\sf{DS}}_0(\pi_{\leq N})>c+\dfrac{\epsilon}{2}$ being $last(\pi_{\leq N})$  a green node (witnessing that $\pi_{\leq N}$ is the prefix of a play $\pi'$ compatible with a best response of Player $1$ to $\sigma^*$, for which $\sf{DS}_0(\pi')>c+\epsilon$).

Moreover, there is a play $\pi\in \out(\sigma^*)$ containing a green
node for which ${\sf{DS}}_1(\pi)\geq d$. This is because of two
reasons. First, a play in $\out(\sigma)$  compatible with a best
response to $\sigma$ by Player $1$ is of the form $hv\pi'$, where $hv$
is a maximal root-to-leaf branch  $b$ of $T$ with $last(b)=v$ green
(by definition of green nodes). Second, for each hystory $hv$ such
that $hv$ is a maximal root-to-leaf branch  $b$ of $T$ with
$last(b)=v$ green,  $\out(\sigma^*)$ contains a play $hv\bar{\pi}$,
where $\bar{\pi}$ is a play starting in $v$ maximizing ${\sf
  DS}_1$. Therefore  ${\sf DS}_1(hv\bar{\pi})= {\sf DS}_1(hv)+\lambda^N {\sf DS}_1(\bar{\pi}) \geq {\sf DS}_1(hv)+\lambda^N {\sf DS}_1(\pi'))=d$, where $hv\pi'$ is a play compatible with a  best response of Player $1$ to $\sigma$. To conclude our proof, we need just to show that each play $\pi\in \out(\sigma^*)$ containing a red node is such that ${\sf{DS}}_1(\pi)<d$. Infact, being $last(\pi_{\leq N})$  red, the history $\pi_{\leq N}$ can not be a prefix of any play in $\out(\sigma)$ compatible with a best response of Player $1$ to $\sigma$. In other words, by playing $\sigma$ Player $0$ allows the adversary to gain a payoff that is at most $r<d$  on each play  $\pi=hv\pi'$ with $v$ red. Therefore, switching her strategy from $\sigma$ to $\sigma^*$ (i.e. playing $\sigma$ for the first $N$ turns and then switching to the memoryless strategy  $\sigma^{{\sf DS}_1}_{min}\in \Sigma_0$) Player $0$ is sure to  let Player $1$   gain a payoff that is at most $r'\leq r<d$ on each play  $\pi=hv\pi'$ with $v$ red.

As a conclusion, against $\sigma^*$ Player $1$ can achieve at least a value $d$. Hence, each best response to $\sigma^*$ visits a green node (if it does not, then ${\sf DS}_1<d$ which is a contradiction). This guarantees that ${\sf DS}_0>c$.


  \end{proof}
  
The  approximation algorithm for solving the Stackelberg values gap problems introduced in Theorem \ref{thm:ASV-CSV-gap} roughly works as follows.  Given a discounted sum game $\mathcal{G}$, a rational threshold $c\in \mathbb{Q}$ and a tolerance rational value $\epsilon>0$, the procedure checks whether there exists a strategy $\sigma_0\in \Sigma_0(\mathcal{G})$ with finite memory $M(\epsilon)$ such that ${\ASV}(\sigma_0)>c$ (resp. ${\CSV}(\sigma_0)>c$ ). If such a strategy exists, the procedure answers {\sf Yes}, otherwise it answers  {\sf No}. 
The correctness of the  outlined procedure follows directly from Lemma \ref{lemmaApprox}.

\begin{theorem}
\label{thm:ASV-CSV-gap}
The gap problems for both the $\CSV$ and $\ASV$ are   decidable for games with discounted-sum objectives.
\end{theorem}
We conclude this subsection by providing a reduction from the partition problem to our gap problems (for both {\CSV} and {\ASV}), showing {\sf NP}-hardness for the corresponding problems.

\begin{figure}[t]
    \centering
\begin{tikzpicture}[->,>=stealth',shorten >=1pt,auto,node distance=2.8cm,
                    thick,scale=0.7,every node/.style={scale=0.8}]
  \tikzstyle{every state}=[fill=gray!30,text=black]
  \tikzstyle{every edge}=[draw=black]
  \tikzstyle{initial}=[initial by arrow]

	\node[thick,state,initial above,initial text=,rectangle,minimum size=25pt]
        (s1) at (0,0) {$v_0$};

	\node[thick,state,minimum size=25pt]
        (s2) at (-3.5,0) {$v_1$};

	\node[thick,state,minimum size=25pt]
        (s3) at (3,0) {$1$};
        
    \node[thick,state,minimum size=25pt]
        (s4) at (6,0) {$2$}; 
        
    \node[thick,state,minimum size=25pt]
        (s5) at (9,0) {$n$};
    
    \node[thick,state,minimum size=25pt]
        (s6) at (12,0) {$v_2$};  
        
        \node        (q_dots) at (7.5,0) {$\cdots$}; 
        
        \draw (s2) edge [loop above] node {(0,0)} node [left,xshift=-0.1cm,yshift=-0.35cm] {} (s2) ;

        \draw (s1) edge [swap] node {$0,T -\dfrac{2}{3}$} node [] {} (s2) ;
        
        \draw (s1) edge node {$(0,0)$} node [] {} (s3) ;
        
        \draw (s3) edge [bend left] node {$(w(1),0)$} (s4) ;
        
        
         \draw (s5) edge [bend left] node {$(w(n),0)$} (s6) ;
         
         \draw (s3) edge [swap, bend right]  node {$(0,w(1))$}  (s4) ;
        
         \draw (s5) edge [swap, bend right]  node {$(0,w(n))$} (s6) ;

          \draw (s6) edge [loop above] node {(0,0)} node [left,xshift=-0.1cm,yshift=-0.35cm] {} (s6) ;




\end{tikzpicture}

\vspace{-2mm}
     \caption{Arena for hardness proof of the gap problem.}
    \label{fig:hardness-dsum}
\vspace{-3mm}
\end{figure}
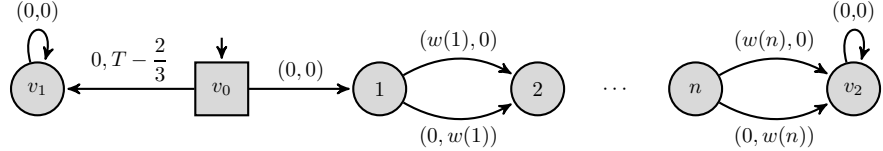

\begin{theorem}
\label{thm:hardness-gap-ds}
The gap problems for both the $\CSV$ and $\ASV$ are {\sf NP}-hard.
\end{theorem}
\begin{proof}

We do a reduction from the Partition problem to our gap problems, working for  both $\CSV$ and $\ASV$.
Let us consider an instance of the partition problem defined by a set $A=\{1,2, \dots n\}$, a function $w : A \rightarrow \mathbb{N}_0$. The partition problem asks if there exists $B \subset A$ such that $\sum_{a \in B} w(a)=\sum_{a \in A \setminus B} w(a)$. W.l.o.g., we assume  $\sum_{a \in A} w(a)=2 \cdot T$ for some $T$. 

To define our reduction, we first fix the two parameters $\lambda \in (0,1)$ and $\epsilon>0$ by choosing values that respect the following two constraints:
\begin{equation}
    \label{eq:D}
    T \cdot \lambda^{n+1} > T-\frac{1}{2}+\epsilon
%
        \quad \quad \quad \quad \quad (T-1) \cdot \lambda^{n+1} < T-\frac{1}{2}-\epsilon
\end{equation}
It is not difficult to see that such values always exist and they can be computed in polynomial time from the description of the partition problem.
Then, we construct the bi-weighted arena ${\cal A}$ depicted in~Fig.~\ref{fig:hardness-dsum}. In this arena, Player~1 has only two choices in the starting state of the game $v_0$. There, he can either send the game to the state $v_1$, and get a payoff of $W-\frac{2}{3}$, 
or he can go to state $1$. 

From state $1$, Player~0 can simulate a partition of the elements of $A$ by choosing edges: left edges simulate the choice of putting the object corresponding to the state in the left class and right edges simulate the choice of putting the corresponding object in the right class. Let $D_0$ and $D_1$ be the discounted sum obtained by Player~0 and Player~1 when arriving in $v_2$. Because $\lambda$ and $\epsilon$ have been chosen according to eq.~\ref{eq:D} 
, we have that:
 $D_0 > W-\frac{1}{2}+\epsilon \land D_1 > W-\frac{1}{2}+\epsilon$
 if and only if
 the choices of edges of Player~0 correspond to a valid partition of $A$.

 Indeed, assume that $B \subseteq A$ is a solution to the partition problem. Assume that Player~0 follows the choices defined by $B$. Then when the game reaches state $b$, the discounted sum of rewards for both players is larger than $W \cdot \lambda^{n+1}$. This is because along the way to $b$, the discounted factor applied on the rewards obtained by both players has always been smaller than $\lambda^{n+1}$ as they were equal to $\lambda^{i+1}$ for all $i \leq n$.  Additionally, we know that sum of (non-discounted) rewards for both players is equal to $W$ as $B$ is a correct partition. Now, it should be clear that both $\ASV(v_0)$ and $\CSV(v_0)$ are greater than $W - \frac{1}{2}+\epsilon$ as in the two cases, Player~1 has no incentive to deviate from the play that goes to $v_1$ as Player~1 would only get $W-\frac{2}{3}$ which is strictly smaller than $D_1$.
 
 Now, assume that there is no solution to the partition problem. In that case, Player~0 cannot avoid to give less than $W-1$ to herself or to Player~1 when going from $v_0$ to $b$. In the first case, its reward is less than $W-1$ and in the second case, the reward of Player~1 is less than $W-1$ and Player~1 has an incentive to deviate to state $v_1$. In the two cases, we have that both $\ASV(v_0)$ and $\CSV(v_0)$ are less than $W-\frac{1}{2}-\epsilon$. 
 So, we have establish that the answer to the gap problem is yes if the partition instance is positive, and the answer is no if the partition instance is negative.

\end{proof}

\bibliographystyle{plain}

\begin{thebibliography}{10}

\bibitem{andersson}
D.~Andersson.
\newblock Improved combinatorial algorithms for discounted payoff games.
\newblock Master's thesis, Department of Information Technology, Uppsala
  University, 2006.

\bibitem{DBLP:conf/lics/BokerHO15}
Udi Boker, Thomas~A. Henzinger, and Jan Otop.
\newblock The target discounted-sum problem.
\newblock In {\em 30th Annual {ACM/IEEE} Symposium on Logic in Computer
  Science, {LICS} 2015, Kyoto, Japan, July 6-10, 2015}, pages 750--761, 2015.

\bibitem{DBLP:conf/lata/BrenguierCHPRRS16}
Romain Brenguier, Lorenzo Clemente, Paul Hunter, Guillermo~A. P{\'{e}}rez,
  Mickael Randour, Jean{-}Fran{\c{c}}ois Raskin, Ocan Sankur, and Mathieu
  Sassolas.
\newblock Non-zero sum games for reactive synthesis.
\newblock In {\em Language and Automata Theory and Applications - 10th
  International Conference, {LATA} 2016, Prague, Czech Republic, March 14-18,
  2016, Proceedings}, volume 9618 of {\em Lecture Notes in Computer Science},
  pages 3--23. Springer, 2016.

\bibitem{DBLP:conf/cav/BrenguierR15}
Romain Brenguier and Jean{-}Fran{\c{c}}ois Raskin.
\newblock Pareto curves of multidimensional mean-payoff games.
\newblock In {\em Computer Aided Verification - 27th International Conference,
  {CAV} 2015, San Francisco, CA, USA, July 18-24, 2015, Proceedings, Part
  {II}}, volume 9207 of {\em Lecture Notes in Computer Science}, pages
  251--267. Springer, 2015.

\bibitem{DBLP:conf/lfcs/BrihayePS13}
Thomas Brihaye, Julie {De Pril}, and Sven Schewe.
\newblock Multiplayer cost games with simple nash equilibria.
\newblock In {\em Logical Foundations of Computer Science, International
  Symposium, {LFCS} 2013, San Diego, CA, USA, January 6-8, 2013. Proceedings},
  volume 7734 of {\em Lecture Notes in Computer Science}, pages 59--73.
  Springer, 2013.

\bibitem{DBLP:conf/concur/ChatterjeeDEHR10}
Krishnendu Chatterjee, Laurent Doyen, Herbert Edelsbrunner, Thomas~A.
  Henzinger, and Philippe Rannou.
\newblock Mean-payoff automaton expressions.
\newblock In {\em {CONCUR} 2010 - Concurrency Theory, 21th International
  Conference, {CONCUR} 2010, Paris, France, August 31-September 3, 2010.
  Proceedings}, volume 6269, pages 269--283. Springer, 2010.

\bibitem{ChattarjeeLPAR2013}
Krishnendu Chatterjee, Vojt{\v{e}}ch Forejt, and Dominik Wojtczak.
\newblock Multi-objective discounted reward verification in graphs and mdps.
\newblock In Ken McMillan, Aart Middeldorp, and Andrei Voronkov, editors, {\em
  Logic for Programming, Artificial Intelligence, and Reasoning}, pages
  228--242, Berlin, Heidelberg, 2013. Springer Berlin Heidelberg.

\bibitem{DBLP:conf/icalp/ConduracheFGR16}
Rodica Condurache, Emmanuel Filiot, Raffaella Gentilini, and
  Jean{-}Fran{\c{c}}ois Raskin.
\newblock The complexity of rational synthesis.
\newblock In {\em 43rd International Colloquium on Automata, Languages, and
  Programming, {ICALP} 2016, July 11-15, 2016, Rome, Italy}, volume~55 of {\em
  LIPIcs}, pages 121:1--121:15. Schloss Dagstuhl - Leibniz-Zentrum fuer
  Informatik, 2016.

\bibitem{DBLP:journals/siamcomp/FerranteR75}
Jeanne Ferrante and Charles Rackoff.
\newblock A decision procedure for the first order theory of real addition with
  order.
\newblock {\em {SIAM} J. Comput.}, 4(1):69--76, 1975.

\bibitem{DBLP:conf/lics/FiliotGR18}
Emmanuel Filiot, Raffaella Gentilini, and Jean{-}Fran{\c{c}}ois Raskin.
\newblock Rational synthesis under imperfect information.
\newblock In {\em Proceedings of the 33rd Annual {ACM/IEEE} Symposium on Logic
  in Computer Science, {LICS} 2018, Oxford, UK, July 09-12, 2018}, pages
  422--431. {ACM}, 2018.

\bibitem{DBLP:conf/tacas/FismanKL10}
Dana Fisman, Orna Kupferman, and Yoad Lustig.
\newblock Rational synthesis.
\newblock In {\em Tools and Algorithms for the Construction and Analysis of
  Systems, 16th International Conference, {TACAS} 2010, Held as Part of the
  Joint European Conferences on Theory and Practice of Software, {ETAPS} 2010,
  Paphos, Cyprus, March 20-28, 2010. Proceedings}, volume 6015 of {\em Lecture
  Notes in Computer Science}, pages 190--204. Springer, 2010.

\bibitem{DBLP:conf/birthday/Goldreich06a}
Oded Goldreich.
\newblock On promise problems: {A} survey.
\newblock In {\em Theoretical Computer Science, Essays in Memory of Shimon
  Even}, volume 3895 of {\em Lecture Notes in Computer Science}, pages
  254--290. Springer, 2006.

\bibitem{DBLP:conf/time/GuptaS14}
Anshul Gupta and Sven Schewe.
\newblock Quantitative verification in rational environments.
\newblock In {\em 21st International Symposium on Temporal Representation and
  Reasoning, {TIME} 2014, Verona, Italy, September 8-10, 2014}, pages 123--131.
  {IEEE} Computer Society, 2014.

\bibitem{DBLP:conf/sefm/GuptaSTDP16}
Anshul Gupta, Sven Schewe, Ashutosh Trivedi, Maram Sai~Krishna Deepak, and
  Bharath~Kumar Padarthi.
\newblock Incentive stackelberg mean-payoff games.
\newblock In {\em Software Engineering and Formal Methods - 14th International
  Conference, {SEFM} 2016, Held as Part of {STAF} 2016, Vienna, Austria, July
  4-8, 2016, Proceedings}, pages 304--320, 2016.

\bibitem{DBLP:journals/corr/GuptaS14c}
Anshul Gupta, Sven Schewe, and Dominik Wojtczak.
\newblock Making the best of limited memory in multi-player discounted sum
  games.
\newblock In {\em Proceedings Sixth International Symposium on Games, Automata,
  Logics and Formal Verification, GandALF 2015, Genoa, Italy, 21-22nd September
  2015}, volume 193 of {\em {EPTCS}}, pages 16--30, 2015.

\bibitem{DBLP:journals/dm/Karp78}
Richard~M. Karp.
\newblock A characterization of the minimum cycle mean in a digraph.
\newblock {\em Discrete Mathematics}, 23(3):309--311, 1978.

\bibitem{DBLP:journals/amai/KupfermanPV16}
Orna Kupferman, Giuseppe Perelli, and Moshe~Y. Vardi.
\newblock Synthesis with rational environments.
\newblock {\em Ann. Math. Artif. Intell.}, 78(1):3--20, 2016.

\bibitem{nash50}
J.~F. Nash.
\newblock Equilibrium points in $n$-person games.
\newblock In {\em PNAS}, volume~36, pages 48--49. National Academy of Sciences,
  1950.

\bibitem{osbornebook}
{Martin J.} Osborne.
\newblock {\em An introduction to game theory}.
\newblock Oxford Univ. Press, 2004.

\bibitem{DBLP:conf/popl/PnueliR89}
Amir Pnueli and Roni Rosner.
\newblock On the synthesis of a reactive module.
\newblock In {\em Conference Record of the Sixteenth Annual {ACM} Symposium on
  Principles of Programming Languages, Austin, Texas, USA, January 11-13,
  1989}, pages 179--190, 1989.

\bibitem{sack1999handbook}
J.-R. Sack and J.~Urrutia, editors.
\newblock {\em Handbook of Computational Geometry}.
\newblock North-Holland Publishing Co., NLD, 2000.

\bibitem{DBLP:journals/iandc/VelnerC0HRR15}
Yaron Velner, Krishnendu Chatterjee, Laurent Doyen, Thomas~A. Henzinger,
  Alexander~Moshe Rabinovich, and Jean{-}Fran{\c{c}}ois Raskin.
\newblock The complexity of multi-mean-payoff and multi-energy games.
\newblock {\em Inf. Comput.}, 241:177--196, 2015.

\end{thebibliography}

\section{Appendix}

\subsection{Proofs of Section \ref{secMP}}

\subsubsection{Proof of Lemma~\ref{lem:emptiness-nonemptiness}}

\begin{proof}
We establish the three statements of this lemma in turn. 
First, consider the arena depicted in Fig.~\ref{fig:game-empty}, where
square nodes are vertices controlled by Player~$1$ while round nodes
are controlled by Player~$0$. We also give names to Player's actions to ease
the description of strategies. In this arena, Player~1 can either always play $a$, noted $\sigma_1^{\omega}$, or play $k$ times $a$ and then $b$, noted $\sigma_1^{k}$. Now, consider the following strategy $\sigma_0$ for Player~0 defined as follows: if Player~1 has played $k$ times $a$ before playing $b$, then play repeatedly $c^k$ followed by one $d$. Clearly, playing $\sigma_1^{\omega}$ has a mean-payoff value of $0$ for Player~2, while playing $\sigma^k_1$ against $\sigma_0$ has a mean-payoff of $\frac{2k+1}{k+1} >0$, so playing $\sigma_1^{\omega}$ is clearly not a best-response to $\sigma_0$. But it is also clear that for all $k_1 < k_2$, we have that $$\MPinf_1(\outcome_{v_0}(\sigma_0,\sigma_1^{k_1})) < \MPinf_1(\outcome_{v_0}(\sigma_0,\sigma_1^{k_2}))$$
\noindent
and so we conclude that there is no best-response for Player~1 to the strategy $\sigma_0$ of Player~0.

Second, as the mean-payoff measure is prefix independent\footnote{In
  the sense that for all $i\in\{0,1\}$, all infinite plays $\pi$ and finite plays
  $\pi'$, $\MPinf_i(\pi'\pi) = \MPinf_i(\pi)$.}, we can w.l.o.g. consider a fixed starting vertex $v_0$ and consider best-responses from there. Let $\sigma_0$ be a finite memory strategy of Player~0 in the arena ${\cal A}$. We note $A(\sigma_0,v_0)$ the finite graph obtained by fixing the strategy $\sigma_0$ for Player~0 in arena ${\cal A}$ from $v_0$. We can consider $A(\sigma_0,v_0)$ as a finite one-player mean-payoff arena as only Player~1 has remaining choices. It is well known that in a finite one-player mean-payoff arena, there are optimal memoryless strategies: they consist in reaching one simple cycle with a maximal mean-payoff~\cite{DBLP:journals/dm/Karp78}. The mean-payoff obtained by such a strategy is the mean-payoff of this cycle and is maximal. A strategy of Player~1 that follows this cycle is thus optimal and it is a best-response to $\sigma_0$.

Third, let $\sigma_0$ be any strategy of Player~0. Let ${\cal
  A}(\sigma_0,v_0)$ denotes the unfolding of the arena ${\cal A}$ for
vertex $v_0$ in which the choices of Player~0 has been fixed by the
strategy $\sigma_0$. We refer to ${\cal A}(\sigma_0,v_0)$ as the tree
$T$ and to any outcome compatible with $\sigma_0$ as an infinite branch $b$ of $T$. Against $\sigma_0$, Player~1 cannot obtain a value which is larger than
$d=\sup_{b \in T} \MPinf_1(b)$. By definition of $\sup$ in the real numbers, for all $\epsilon > 0$, there exists a branch $b \in T$ such that $\MPinf_1(b) > d-\epsilon$. So, for every $\epsilon > 0$, there exists a branch $b$ and a strategy $\sigma_1^b$ that follows this branch against $\sigma_0$ and which is thus an  $\epsilon$-best-response.
\end{proof}

\subsubsection{Proof of Theorem~\ref{thm:equiv-def}}

\begin{proof}
First, consider the game depicted in Fig.~\ref{fig:no-opt-ASV}. First let us show that $\ASV(v_0)=1$. 
For all $\epsilon >0$, assume that Player~0 plays $\sigma_0^{k(\epsilon)}$ defined as: repeat forever, from $v_1$ play one time $v_1 \rightarrow v_1$ and then repeat playing $v_1 \rightarrow v_0$ for $k(\epsilon)$ times, with $k$ chosen such that the mean-payoff for Player~0 is larger than $1-\epsilon$. Such a $k$ always exists. The best-response of Player~1 to $\sigma_0^{k(\epsilon)}$ is to always play $v_0 \rightarrow v_1$ as by playing this edge forever, Player~1 gets a mean-payoff strictly larger than $1$. 
Clearly, by playing less frequently $v_1 \rightarrow v_1$, Player~0 can obtain a value which is arbitrary close to $1$. But in addition, we note that the only way for Player~0 to reach value $1$ would be to play $v_1 \rightarrow v_1$ with a frequency that goes to $0$ in the limit. And in that case, the mean-payoff obtained by Player~1 would be equal to $1$. So it would not be better than the mean-payoff that Player~1 gets when playing $v_1 \rightarrow v_2$. As a consequence, in that case $v_1 \rightarrow v_2$ would also be a best-response too, and the adversarial Stackelberg value of that strategy of Player~0 would be equal to $0$.

Second, we show the following equivalence, which directly implies the second part of the theorem
(by taking $c = \ASV(v)-\epsilon$): $\ASV(v) > c$ iff
$\exists \sigma_0\in \Sigma_0\cdot \ASV(\sigma_0)(v) > c$.

Let us now prove this equivalence. By definition of
$\ASV(\sigma_0)(v)$, we have to show that 
$$\ASV(v)> c \text{ iff }\exists \sigma_0 \cdot  \exists \tau >0 :
\BR_1^{\tau}(\sigma_0)\neq\varnothing \land \forall \sigma_1 \in
\BR_1^{\tau}(\sigma_0) : \MPinf_0(\outcome_v(\sigma_0,\sigma_1)) >
c.$$

\noindent
The right to left direction is trivial as $\sigma_0$ and $\tau$ can play the role of witness of $\ASV(v) > c$. 

For the left to right direction, let $c' = \ASV(v)$. By definition of
$\ASV(v)$, we have
$$c'=\sup_{\sigma_0,\epsilon \geq 0 ~\mid~ \BR_{1}^{\epsilon}(\sigma_0)\neq\varnothing} ~~~\inf_{\sigma_1 \in \BR_{1}^{\epsilon}(\sigma_0)} \MPinf_0(\outcome_v(\sigma_0,\sigma_1)) > c$$
\noindent
By definition of $\sup$, for all $\delta >0$, we have that:
$$ \exists \sigma^{\delta}_0 \cdot \exists \epsilon^{\delta}>0 :  
\BR_1^{\epsilon^{\delta}}(\sigma^{\delta}_0)\neq\varnothing \land \inf_{\sigma_1 \in  \BR_1^{\epsilon^{\delta}}(\sigma^{\delta}_0)} \MPinf_0(\outcome_v(\sigma^{\delta}_0,\sigma_1)) \geq c'-\delta $$
\noindent
which in turn implies:
$$ \exists \sigma^{\delta}_0 \cdot \exists \epsilon^{\delta}>0 :  
\BR_1^{\epsilon^{\delta}}(\sigma^{\delta}_0)\neq\varnothing \land \forall \sigma_1 \in  \BR_1^{\epsilon^{\delta}}(\sigma^{\delta}_0): \MPinf_0(\outcome_v(\sigma^{\delta}_0,\sigma_1)) \geq c'-\delta $$
Now let us consider $\delta >0$ such that $c' - \delta > c$. Such a $\delta$ exists as $c' > c$. Then we obtain:
$$\exists \sigma_0 \cdot  \exists \tau >0 : \BR_1^{\tau}(\sigma_0)\neq\varnothing \land \forall \sigma_1 \in  \BR_1^{\tau}(\sigma_0) : \MPinf_0(\outcome_v(\sigma_0,\sigma_1)) > c.$$

\noindent
Finally, we note that the need for memory for $\epsilon$ approximation is a consequence of the example used in Theorem~\ref{thm:equiv-def}.
\end{proof}

\subsubsection{Proof of Lemma~\ref{lem:smallwitness}}

\begin{proof}
We first give a name to simple paths that are useful in the proof. First, the simple path that goes from from $v$ to $\ell_1$, is called $\pi_1$. Second, the simple path that goes from $\ell_1$ to $\ell_2$ is called $\pi_2$. Finally, the simple path that goes from $\ell_2$ back to $\ell_1$ is called $\pi_3$.

The right to left direction consists in showing that the existence of
$\pi_1$, $\pi_2$, $\pi_3$ and of the two simple cycles $\ell_1$, $\ell_2$
implies the existence of a witness for $\ASV(v) > c$ as required by
Theorem~\ref{thm:witness-mp}. For all $i\in\mathbb{N}\setminus\{0\}$,
we let $\rho_i = \ell_1^{\lceil \alpha \cdot i \rceil} \cdot \pi_2 \cdot
\ell_2^{\lceil \beta \cdot i \rceil} \cdot \pi_3$ and define the witness
$\pi$ as follows:
$$\pi=\pi_1 \rho_1\rho_2\rho_3\dots
$$
\noindent
It is easy to show that $\MPinf_0(\pi)=\alpha \cdot w_0(\ell_1) + \beta
\cdot w_0(\ell_2)$ which is greater than $c$ by hypothesis. Indeed, by
construction of $\pi$, we have that the importance of the non-cyclic
part ($\pi_2$ and $\pi_3$) is vanishing as $i$ is getting large, and
as the mean-payoff measure is prefix independent, the role of $\pi_1$
can be neglected. For the same reason, we have that
$\MPinf_1(\pi)=\alpha \cdot w_1(\ell_1) + \beta \cdot w_1(\ell_2)$ which is
equal to $d$ by hypothesis. It remains to show that $\pi$ does not
cross a $(c,d)$-bad vertex. This is direct by construction of $\pi$
and point $(3)$.

Let us now consider the left to right direction. Let $\pi$ be a
witness for $\ASV(v) >c$. By Theorem~\ref{thm:witness-mp}, we have
that $\pi$ starts in $v$, $\MPinf_0(\pi)=c' > c$ and
$\MPinf_1(\pi)=d$, and all the vertices $v'$ along $\pi$ are such that
$v' \nvDash \Estrat{1} \MPinf_0 \leq c \land \MPinf_1 \geq d$. Let us
note $D$ the set of vertices that appears infinitely often along
$\pi$. As the set of vertices is finite, we know that $D$ is non-empty
and there exists an index $i \geq 0$ such that the states visited
along $\pi(i\dots)$ is exactly those vertices in $D$. So, clearly, in
the graph underlying the game arena, $D$ is a strongly connected
component. So, in the strongly connected component $D$, there is an
infinite play $\pi'$ that is such that  $\MPinf_0(\pi')=c' > c$ and
$\MPinf_1(\pi')=d$. According to Proposition~\ref{prop:larger} and Theorem~\ref{thm:caractSCC},
if in a strongly connected component, there is a play with
$\MPinf_0(\pi')=c' > c$ and  $\MPinf_1(\pi')=d$ there is a convex
combination of coordinates of simple cycles that gives a value $(x,y)$
such that $x \geq c'$ and $y \geq d$. As we are concerned here with
mean-payoff games with $2$ dimensions, we can apply the CarathÃ©odory
baricenter theorem to deduce that there exists a set of cycles of
cardinality at most $3$, noted $\{ \ell_{i_1}, \ell_{i_2}, \ell_{i_3}\}$, and
$\alpha_{i_1}$, $\alpha_{i_2}$, and $\alpha_{i_3}$, such that  the
convex hull of vectors $(w_0(\ell_{i_1}),w_1(\ell_{i_1}))$,
$(w_0(\ell_{i_2}),w_1(\ell_{i_2}))$, and $(w_0(\ell_{i_3}),w_1(\ell_{i_3}))$
intersects with the set $P=\{(x,y) \mid x > c \land y \geq d \}$. This
convex hull is a triangle. If a triangle intersects $P$, it has to be the case that one of its edges intersects $P$. This edge is definable as the convex combination of two of the vertices of the triangle. As a consequence, we conclude that there are two simple cycle $\ell_1$ and $\ell_2$, and two rational values $\alpha,\beta \geq 0$ such that $\alpha+\beta=1$ and $\alpha \cdot w_0(\ell_1)+ \beta \cdot w_0(\ell_2) > c$ and $\alpha \cdot w_1(\ell_1)+ \beta \cdot w_1(\ell_2) \geq d$. It remains to show how to construct $\pi_1$, $\pi_2$, and $\pi_3$. For $\pi_1$, we concatenate a play without repetition of vertices from vertex $v$ to the set $D$ that only takes vertices also in $\pi$, then when in $D$, we take a finite play without repetition to the simple cycle $\ell_1$. For $\pi_2$, we take a simple play from the cycle $\ell_1$ to the cycle $\ell_2$, and fr $\pi_3$, we take a simple play from $\ell_2$ to $\ell_1$. The existence of those plays is guaranteed by the fact that $D$ is strongly connected. It is easy to verify that the constructed plays and cycles have all the properties required.      
\end{proof}

\subsubsection{Proof of Lemma~\ref{lem:badval-effective}}

To establish this lemma, we start from results that have been
established in~\cite{DBLP:conf/cav/BrenguierR15} where the Pareto
curve of $d$-dimensional mean-payoff games is studied.

\subparagraph{Pareto curve}  Before giving the formal details, we recall the notion of {\em Pareto curve} associated with a $d$-dimensional mean-payoff game. Those games are played between two players, called here \eve\/ and \adam, on a game arena where each edge is labelled by a $d$-dimensional vector of weights. In such context, we are interested in strategies of \eve\/ that ensure thresholds as high as possible on all dimensions. However, since the weights are multidimensional, there is not a unique maximal threshold in general. The concept of {\em Pareto optimality} is used to identify the most interesting thresholds.
To define the set of Pareto optimal thresholds, we first define the set of thresholds that \eve\/ can force from a vertex $v$:
$${\sf Th}(\calG,v) \\ =\left\{ x \in \mathbb{R}^d \mid \exists \sigma_\exists \cdot \forall \pi \in \outcome_v(\sigma_\exists) \cdot \forall i : 1 \leq i \leq d : \MPinf_i(\pi) \ge x_i \right\}.$$
  A threshold $c \in \mathbb{R}^{d}$ is \newdef{Pareto optimal} from
  $v$ if it is maximal in the set ${\sf Th}(\calG,v)$. So the set of Pareto optimal thresholds is defined as:
  $$\PO(\calG,v)= \{  x \in {\sf Th}(\calG,v) \mid \neg\exists x' \in
  {\sf Th}(\calG,v): x' > x \}\qquad \text{(for the component-wise order)}$$
We refer to this set as the \emph{Pareto curve} of the game.
Note that the set of thresholds that \eve\/ can force is exactly equal
to the downward closure, for the component-wise order, of the Pareto optimal thresholds, i.e. ${\sf Th}(\calG,v)= \downarrow \PO(\calG,v)$.

\subparagraph{Cells} We recall here the notion of cells in geometry,
which is useful to represent the set
of Pareto optimal thresholds. Let $a \in \mathbb{Q}^d$ be a vector in $d$ dimensions.
  The associated \newdef{linear function} $\alpha_a \colon \mathbb{R}^d \mapsto \mathbb{R}$ is the function $\alpha_a (x) = \sum_{i\in \lsem 1,d\rsem} a_i \cdot x_i$ that computes the weighted sum relative to $a$.
  A \newdef{linear inequation} is a pair $(a,b)$ where $a\in \mathbb{Q}^d \setminus \{ \vec0\}$ and $b \in \mathbb{Q}$.
  The \newdef{half-space} satisfying $(a,b)$ is the set $\ssp(a,b) = \{ x \in \mathbb{R}^d \mid  \alpha_a(x) \ge b \}$.
  A \newdef{linear equation} is also given by a pair $(a,b)$ where $a\in \mathbb{Q}^d \setminus \{ \vec0\}$ and $b \in \mathbb{Q}$ but we associate with it the \newdef{hyperplane} 
  $\hp(a,b) = \{ x \in \mathbb{R}^d \mid  \alpha_a(x) = b \}$.
  If $H = \ssp(a,b)$ is a half-space, we sometimes write $\hp(H)$ for the \emph{associated hyperplane} $\hp(a,b)$.
  A \newdef{system of linear inequations} is a set $\lambda = \{(a_1,b_1),\dots,(a_l,b_l)\}$ of linear inequations.
  The \emph{polyhedron} generated by $\lambda$ is the set $\pol(\lambda) = \bigcap_{(a,b)\in\lambda} \ssp(a,b)$.

We say that two points $x$ and $y$ are equivalent with respect to a
set of half-spaces $\mathcal{H}$, written $x \sim_{\mathcal{H}} y$, if
they satisfy the same set of equations and inequations defined by
$\mathcal{H}$. Formally $x \sim_{\mathcal{H}} y$ if for all $H\in
\mathcal{H}$, $x \in H \Leftrightarrow y \in H$ and $x \in \hp(H)
\Leftrightarrow y \in \hp(H)$. Given a point $x$, we write
$[x]_{\mathcal{H}} = \{ y \mid x \sim_{\mathcal{H}} y \}$ the
equivalence class of $x$. These equivalence classes are known in
geometry as \emph{cells}~\cite{sack1999handbook}. We write
$C(\mathcal{H})$ the set of cells defined by $\mathcal{H}$. Cells can
be represented as a disjunction of conjunctions of strict and non-strict linear
inequations of the form $\sum_{i=1}^d a_i.x_i>b$ and $\sum_{i=1}^d
a_i.x_i\geq b$ respectively.

\begin{theorem}[\cite{DBLP:conf/cav/BrenguierR15}]\label{thm:val}
  There is a deterministic exponential algorithm that given a $d$-dimensional mean-payoff game $\calG$ and a vertex $v$ computes an effective representation of $\PO(\calG,v)$ and $\val(\calG,v)$ as a union of cells or equivalently as a formula of the theory of the reals $\langle \mathbb{R},+,\geq\rangle$ with $d$ free variables.
  Moreover, when the dimension $d$ is fixed and the weights are polynomially bounded then the algorithm works in deterministic polynomial time.
\end{theorem}

While the set $\Lambda(v)$ is not equal to ${\sf Th}(\calG,v)$, the
definition of the two sets are similar in nature (in dimension 2), and
we show next that we can adapt the algorithm
of~\cite{DBLP:conf/cav/BrenguierR15} used to prove Theorem~\ref{thm:val} to compute a symbolic representation of $\Lambda(v)$. For that we rely on two propositions. Each proposition deals with one of the differences in the definitions of $\Lambda(v)$ and ${\sf Th}(\calG,v)$. First, while in ${\sf Th}(\calG,v)$, \eve\/ aims at maximizing the mean-payoff in each dimension, for $\Lambda(v)$, Player~1 wants to minimize the mean-payoff on the first dimension (the payoff of Player~0) and maximize the mean-payoff on the second dimension (his own Player~1 payoff).  But this discrepancy can be easily handled using the following property:
      
      \begin{proposition}
      \label{prop:mp-inverse}
    For all mean-payoff games $\game$, for all play $\pi \in {\sf Play}_{\game}$, for all thresholds $c \in \mathbb{R}$, we have that $\MPinf_1(\pi) \leq c$ if and only if $-\MPsup_1(\pi) \geq -c$.
    \end{proposition}
 
So, if we inverse the payoff on the first dimension on each edge of the game, we end up with an equivalent game where Player~1 now wants to maximize the value he can obtain on the two dimensions. We note $\game'$ this new game. The only remaining difference that we need to deal with in order to be able to apply Theorem~\ref{thm:val} is that one of the dimension is now measured by a {\it mean-payoff sup} and not {\it mean-payoff inf}. The following proposition tells us that we can safely replace the limsup mean-payoff by a liminf mean-payoff. This is a direct corollary of a more general statement in~\cite{DBLP:journals/iandc/VelnerC0HRR15}:

\begin{proposition}[Lemma~14~in~\cite{DBLP:journals/iandc/VelnerC0HRR15}]
\label{prop:mp-inf-sup}
For all mean-payoff games $\game$, for all state $v \in V$, for all $c,d \in \mathbb{Q}$, $$v \models \Estrat{1} -\MPsup_0 \geq -c \land \MPinf_1 \geq d$$ 
\noindent 
if and only if $$v \models \Estrat{1} -\MPinf_0 \geq -c \land \MPinf_1 \geq d.$$
\end{proposition}
\noindent

By Proposition~\ref{prop:mp-inverse} and
Proposition~\ref{prop:mp-inf-sup}, we have shown that it suffices to
inverse the weights of the first dimension in the bi-weighted graph to
obtain a two-dim. mean-payoff game in which the set of thresholds
${\sf Th}(\calG,v)$ that $\eve$ can enforce is exactly the set
$\Lambda(v)$. As a consequence, we can use the algorithm behind
Theorem~\ref{thm:val} to compute a symbolic representation (in the theory $\langle \mathbb{R},+, \leq \rangle$) of this set
in {\sf ExpTime}, achieving to prove Lemma~\ref{lem:badval-effective}.

\subsection{Proofs of Section \ref{secDS}}

\subparagraph{Proof of Theorem \ref{thm:ASV-CSV-gap}}
\begin{proof}
 By Lemma \ref{lemmaApprox}, if $\mathcal{G}\in {\sf Yes}^{\epsilon,c}_\ASV$  then Player $0$ has a strategy  with finite memory $M(\epsilon)$ such that for all  $\tau\in BR_1(\sigma^*)$   it holds $\sf{DS}^\lambda_0({\sf Out}(\sigma^*,\tau) > c$. Therefore, to solve the
\ASV-gap problem,   with gap $\epsilon$ and  threshold $c$ it is sufficient to apply the following two-steps procedure: 
\begin{enumerate}
    \item check if there exists a strategy $\sigma\in \Sigma_0(\mathcal{G})$ with finite memory $M(\epsilon)$ such that $\ASV(\sigma)>c$.
    \item  If such a strategy exists answer YES, otherwise answer  NO.
\end{enumerate} 
In particular, given a finite memory strategy $\sigma$ for Player $0$, checking weather $\ASV(\sigma)>c$ can be done by first computing the product of the game $\mathcal{G}$ with the finite memory strategy $\sigma$. This yields to a single Player game $\mathcal{G}'$ (controlled by Player $1$). By \cite{andersson}, the problem of solving a one-player (single-valued) discounted-sum game   can be stated as a linear program. Therefore, we can use linear programming to determine, for each vertex $x$ of $\mathcal{G}'$, the maximal discounted  sum $\DS_1(\pi)$ that Player $1$ can obtain by following a path in $\mathcal{G}'(x)$. At this point, we can  delete from $\mathcal{G}'$ each edge that is not consistent with the solution of the above  linear program, obtaining a new (single player) discounted sum games $\mathcal{G}''$ where each play is consistent with both $\sigma$ and a best response of Player $1$ to $\sigma$. Answering weather $\ASV(\sigma)>c$ finally amounts to determine the minimal discounted  sum $\DS_0(\pi)$ that Player $1$ can obtain by following a path in $\mathcal{G}''(v)$, i.e. solving again a one player single value discounted sum game problem.

We conclude by  showing that our two-steps procedure answers {\sf Yes} if  $\mathcal{G}\in {\sf Yes}^\epsilon_c$,  answers {\sf No} if  ${\sf No}^\epsilon_c$,  answers arbitrarly otherwise. Assume  $\mathcal{G} \in {\sf Yes}^\epsilon_c$. Then, by Lemma \ref{lemmaApprox} Player $0$ has a strategy with finite memory $M(\epsilon)$ witnessing $\ASV>c$ and the algorithm answers {\sf Yes}. If   $\mathcal{G}\in {\sf No}^\epsilon_c$,  Player $0$ has no strategy  witnessing $\ASV>c$. Therefore  Player $0$ has no  finite memory strategy  witnessing $\ASV>c$ and the algorithm answers {\sf No}.  The answer of the algorithm is not guaranteed to be neither {\sf Yes} nor {\sf No} if $\mathcal{G}\notin $ ${\sf Yes}^\epsilon_c\cup {\sf No}^\epsilon_c$.

 The \CSV-gap  algorithm is similarly defined on the ground of Lemma \ref{lemmaApprox}. In particular,  to solve the
\CSV-gap problem,   with gap $\epsilon$ and  threshold $c$ it is sufficient to  proceed as follows: 
Check if there exists a strategy $\sigma\in \Sigma_0(\mathcal{G})$ with finite memory $M(\epsilon)$ such that $\CSV(\sigma)>c$. If such a strategy exists answer YES, otherwise answer  NO.

\end{proof}

\end{document}